\documentclass[a4paper,12pt]{article}
\usepackage{amsfonts,slashed}
\usepackage{url}
\usepackage{latexsym}
\usepackage{amsfonts}
\usepackage{epsfig}
\usepackage{latexsym,amssymb}
\usepackage{amsmath,amssymb,amsthm}

\usepackage{ifpdf}
\usepackage{cite}
\usepackage[pdftex]{hyperref}
\usepackage{tikz}
\ifx\pdfoutput\undefined
   \pdffalse
\else
   \pdfoutput=1
   \pdftrue
\fi

\numberwithin{equation}{section}
\def\be{\begin{equation}}
\def\ee{\end{equation}}
\def\ba{\begin{array}}
\def\ea{\end{array}}

\newcommand{\bea}{\begin{eqnarray}}
\newcommand{\eea}{\end{eqnarray}}

\newcommand{\bbox}{\lower.2ex\hbox{$\Box$}}

\newtheorem{theorem}{Theorem} 
\newtheorem{definition}{Definition} 
\def\bfone{\relax{\rm 1\kern-.35em 1}}
\def\bfzero{\relax{\rm 0\kern -.45 em 0}}
\usetikzlibrary{trees}
\usetikzlibrary{positioning,shadows,arrows,decorations.pathmorphing,decorations.markings}
\tikzstyle{block}=[draw opacity=0.7,line width=1.4cm]

\begin{document}

\title{Infinite $S$-Expansion with Ideal Subtraction and Some Applications}
\author{D. M. Pe\~{n}afiel$^{1,2,3}$\thanks{%
diegomolina@udec.cl}, L. Ravera$^{2,3}$\thanks{%
lucrezia.ravera@polito.it} \\
{\small $^{1}$\textit{Departamento de F\'{\i}sica, Universidad de Concepci\'{o}n,}} \\
{\small Casilla 160-C, Concepci\'{o}n, Chile}\\
{\small $^{2}$\textit{DISAT, Politecnico di Torino}}\\
{\small Corso Duca degli Abruzzi 24, I-10129 Torino, Italia}\\
{\small $^{3}$\textit{Istituto Nazionale di Fisica Nucleare (INFN)}}\\
{\small Sezione di Torino, Via Pietro Giuria 1, 10125, Torino, Italia}}
\maketitle

\vskip 1 cm

\begin{center}
{\small \textbf{Abstract} }
\end{center}

According to the literature, the $S$-expansion procedure involving a finite semigroup is valid no matter what the structure of the original Lie (super)algebra is; However, when something about the structure
of the starting (super)algebra is known and when certain particular conditions are met, the $S$-expansion method (with its features of resonance and reduction) is able not only to lead to several kinds of expanded (super)algebras, but also to reproduce the effects of the standard as well as the generalized In\"on\"u-Wigner contraction.
In the present paper, we propose a new prescription for $S$-expansion, involving 
an infinite abelian semigroup $S^{(\infty)}$ and the subtraction of an infinite ideal subalgebra. We show that the subtraction of the infinite ideal subalgebra corresponds to a reduction. 
Our approach is a generalization of the finite $S$-expansion procedure presented in the literature, and it offers an alternative view of the generalized In\"on\"u-Wigner contraction.  
We then show how to write the invariant tensors of the target (super)algebras in terms of those of the starting ones in the infinite $S$-expansion context presented in this work. We also give some interesting examples of application on algebras and superalgebras.

\vskip 1 cm


\vskip 1 cm \eject
\numberwithin{equation}{section}


\section{Introduction}

The contraction of finite dimensional Lie algebras is a tool that became very well known thanks to the papers of In\"on\"u and Wigner (see Refs. \cite{IW1, Inonu}), in which the authors developed
the so-called \textit{In\"on\"u-Wigner contraction}. This contraction procedure has then been successfully used in order to obtain different (super)algebras from given ones. 
Recently, a generalization of the In\"on\"u-Wigner contraction was obtained in Ref. \cite{lastpat}, by rescaling not only the generators of a Lie superalgebra, but also the arbitrary constants appearing in the components of the invariant tensor.

On the other hand, a method for expanding Lie algebras (\textit{Lie algebras expansion}, also known as \textit{power series expansion}) was introduced in Ref. \cite{Hatsuda} and then studied and applied in diverse works to algebras and superalgebras (see Refs. \cite{Azcarraga1,Azcarraga2,Azcarraga13} for further details).
Subsequently, an alternative to the method of power series expansion, called \textit{$S$-expansion} procedure, was developed (see Refs. \cite{Izaurieta,Izaurieta2,Iza}) and applied in diverse scenarios, both in Mathematics and in Physics.

The $S$-expansion method replicates through the elements of a semigroup the structure of the original algebra into a new one. 
The basis of the $S$-expansion consists, in fact, in combining the inner multiplication law of a discrete set $S$ with the structure of a semigroup, with the structure constants of a Lie algebra $\mathfrak{g}$; The new, larger Lie algebra thus obtained is called \textit{$S$-expanded algebra}, and it is written as $\mathfrak{g}_S= S \times \mathfrak{g}$. 
An important goal of the $S$-expansion procedure is that it allows to write the invariant tensor of the $S$-expanded algebra from the knowledge of the invariant tensor of the original one. 

From the physical point of view, several (super)gravity theories have been extensively studied using the $S$-expansion approach, enabling numerous results over recent years (see Refs. \cite{Iza4, GRCS, CPRS1, Topgrav, BDgrav, CR2, CRSnew, Static, Gen, Ein, Fierro1, Fierro2, Kondrashuk, Artebani, Concha1, Diego, Salgado, Concha2, Caroca:2010ax, Caroca:2010kr, Caroca:2011zz, Andrianopoli:2013ooa, Concha:2016hbt, Concha:2016kdz, Concha:2016tms, Durka:2016eun}). 
Recently, in Ref. \cite{Marcelo2}, an analytic method for $S$-expansion was developed. This method gives the multiplication table(s) of the (abelian) set(s) involved in an $S$-expansion process for reaching a target Lie (super)algebra from a starting one, after having properly chosen the partitions over subspaces of the considered (super)algebras.

There are two facets applicable in the $S$-expansion method, which offer great manipulation on algebras, \textit{i.e.} \textit{resonance} (that transfers the structure of the semigroup to the target algebra, allowing to control its structure with a suitable choice on the semigroup decomposition) and \textit{reduction} (which plays a peculiar role in cutting the algebra properly).

The $S$-expansion procedure involving a \textit{finite} semigroup is valid no matter what the structure of the original Lie (super)algebra is. However, when something about the structure
of the starting (super)algebra is known and when certain particular conditions are met, the $S$-expansion method is able not only to lead to several kinds of expanded (super)algebras, but also to reproduce the effects of the standard as well as the generalized In\"on\"u-Wigner contraction (the latter can also be referred to as ``Weimar-Woods generalized contraction" and involves higher powers of the contraction parameter).
In Ref. \cite{Izaurieta}, the authors explained how information on the subspace structure of the original algebra can be used in order to find resonant subalgebras of the $S$-expanded
algebra and discussed how this information can be put to use
in a different way, namely, by extracting reduced algebras from
the resonant subalgebra. By following this path, the generalized In\"on\"u-Wigner contraction fits within their scheme, and can thus be reproduced by $S$-expansion.

In the present paper, we develop a new prescription for the $S$-expansion procedure, which involves an infinite abelian semigroup and the removal of an (infinite) ideal subalgebra. We will call this procedure \textit{infinite $S$-expansion}; It represents an extension and generalization of the 
finite one and it offers an alternative view of the generalized In\"on\"u-Wigner contraction. Indeed, we will explicitly show how to reproduce a generalized In\"on\"u-Wigner contraction by following our approach.
We then explain how to reconstruct the components of the invariant tensor of the infinitely $S$-expanded algebra with ideal subtraction from those of the initial one. This is useful since it allows to develop the dynamics and construct the Lagrangians of physical theories, starting from their algebraic structure.

Our alternative method differs from the one presented in Ref. \cite{Izaurieta} (that is the finite $S$-expansion), since it involves an infinite semigroup and the subtraction of an infinite ideal subalgebra. Obviously, the use of an infinite semigroup is not the only way for reaching the results we present in this paper (in fact, the examples of application we will consider are already well known from the literature). However, in the present paper we wish to 
describe an alternative path to expansion and contraction, in which, interestingly, the same results obtained in the context of finite $S$-expansion and In\"on\"u-Wigner contraction (where, in particular, the latter does not change the dimension of the original (super)algebra) can be reproduced by acting on the original (super)algebra with an infinite abelian semigroup and by consequently subtracting an infinite ideal. 
The subtraction of the infinite ideal subalgebra here is crucial, since it allows to obtain Lie (super)algebras with a finite number of generators, after having infinitely expanded the original Lie (super)algebras. 

Furthermore, we will see that the prescription for infinite $S$-expansion with ideal subtraction leads to reduced algebras, since the subtraction of the infinite ideal can be viewed as a reduction involving an infinite number of semigroup elements that, together with the generators associated, play the role of ``generating zeros". In this sense, the subtraction of the infinite ideal can be viewed as a generalization of the $0_S$-reduction. Let us observe that the role of ``generating zeros" in the $0_S$-reduction is played by the zero element, while in our method is the whole ideal subalgebra that plays this role; Thus, the concept, in this sense, is slightly different. However, we will show that our method reproduces the same result of a $0_S$-reduction.

This paper is organized as follows: In Section \ref{Review}, we give a review of the $S$-expansion procedure, with reduction and resonance, and of the In\"on\"u-Wigner contraction. We also extend the concept of normal subgroup of a group to the concept of ideal subalgebra of an algebra, since the latter will be useful in the rest of this work. In Section \ref{Generalized}, we briefly review the way in which a standard In\"on\"u-Wigner contraction can be reproduced with a finite $S$-expansion procedure (with resonance and $0_S$-reduction). Then, we discuss how the authors of \cite{Izaurieta} used the information on the subspace structure of the original Lie algebra to extract reduced algebras from resonant subalgebras of $S$-expanded algebras, being able to reproduce, in this way, the generalized In\"on\"u-Wigner contraction. Subsequently, we describe our prescription for infinite $S$-expansion with ideal subtraction and we explain how the generalized In\"on\"u-Wigner contraction is reproduced in this context. Our method consists in first of all performing an $S$-expansion procedure involving an infinite abelian semigroup (which we will denote by $S^{(\infty)}$), and in consequently subtracting an (infinite) ideal subalgebra to the infinitely $S$-expanded one. 
We explicitly show that the subtraction of the infinite ideal subalgebra corresponds to a reduction. 
We also explain how to write the invariant tensors of the target (super)algebras in terms of those of the original ones in this context.
In Section \ref{Examples}, we give some example of application on different (super)algebras, reproducing some results presented in the literature in the context of finite expansions and contractions, and we write the invariant tensors of some of the mentioned (super)algebras in terms of the invariant tensors of the starting ones. Section \ref{Comments} contains a summary of our results, with comments and possible developments.

\section{Review of $S$-expansion, In\"on\"u-Wigner contraction, and ideal subalgebras} \label{Review}

The $S$-expansion of Lie (super)algebras through abelian semigroups consists in considering the direct product between an abelian semigroup $S$ and a Lie algebra $\mathfrak{g}$: $S \times \mathfrak{g}$.

Under general conditions, relevant subalgebras can be systematically extracted from $S\times\mathfrak{g}$, like resonant algebras and reduced ones. 

In the following, we give a review of the $S$-expansion, reduction, and resonance procedures \cite{Izaurieta} and of the In\"on\"u-Wigner contraction process \cite{IW1, Inonu}. 

We also extend the concept of normal subgroup of a group to the concept of ideal subalgebra of an algebra, which will be useful in the development of the (infinite) $S$-expansion prescription described in this paper.

\subsection{$S$-expansion for an arbitrary semigroup $S$}

The $S$-expansion procedure \cite{Izaurieta} consists in combining the structure constants of a Lie algebra $\mathfrak{g}$ with the inner multiplication law of a semigroup $S$, in order to define the Lie bracket of a new, $S$-expanded algebra $\mathfrak{g}_S= S \times \mathfrak{g}$. 

\begin{definition}\label{defSexp} 
Let $S= \lbrace \lambda_\alpha \rbrace$, with $\alpha=1,...,N$, be a finite, abelian semigroup with two-selector $K_{\alpha \beta}^{\;\;\;\; \gamma}$ defined by
\begin{equation}\label{kseldef}
K_{\alpha \beta}^{\;\;\;\; \gamma} = \left\{ \begin{aligned} &
1 , \;\;\;\;\; \text{when} \; \lambda_\alpha \lambda_\beta = \lambda_\gamma,
\\ & 0 , \;\;\;\;\; \text{otherwise}. \end{aligned} 
\right.
\end{equation}
Let $\mathfrak{g}$ be a Lie algebra with basis $\lbrace T_A \rbrace$ and structure constants $C_{AB}^{\;\;\;\;C}$, defined by the commutation relations 
\begin{equation}
\left[T_A, T_B \right]= C_{AB}^{\;\;\;\;C}\; T_C .
\end{equation}
Denote a basis element of the direct product $S\times \mathfrak{g}$ by $T_{(A,\alpha)}= \lambda_\alpha T_A$, and consider the induced commutator
\begin{equation}
\left[ T_{(A,\alpha)},T_{(B,\beta)}\right] \equiv \lambda_\alpha \lambda_\beta \left[T_A,T_B \right].
\end{equation}
One can show \cite{Izaurieta} that the product
\begin{equation}\label{prodsexp}
\mathfrak{g}_S= S \times \mathfrak{g}
\end{equation}
corresponds to the Lie algebra given by
\begin{equation}\label{expandedone}
\left[T_{(A,\alpha)},T_{(B,\beta)}\right]= K_{\alpha \beta}^{\;\;\;\; \gamma} C_{AB}^{\;\;\;\;C}\; T_{(C,\gamma)},
\end{equation}
whose structure constants can be written as
\begin{equation}\label{strconstant}
C_{(A,\alpha)(B,\beta)}^{\;\;\;\;\;\;\;\;\;\;\;\;\;\;\;\;(C,\gamma)}= K_{\alpha \beta}^{\;\;\;\gamma}C_{AB}^{\;\;\;\;C}.
\end{equation}
For every abelian semigroup $S$ and Lie algebra $\mathfrak{g}$, the algebra $\mathfrak{g}_S$ obtained through the product (\ref{prodsexp}) is also a Lie algebra, with a Lie bracket given by (\ref{expandedone}).
The new, larger Lie algebra thus obtained is called \textit{$S$-expanded algebra}, and it is written as $\mathfrak{g}_S= S \times \mathfrak{g}$. 
\end{definition}

\subsection{Reduced algebras}

Let us give the following definition (see Ref. \cite{Izaurieta}) in order to introduce the concept of reduction of Lie algebras:

\begin{definition}\label{defreduced}
Let us consider a Lie algebra $\mathfrak{g}$ of the form $\mathfrak{g}=V_0 \oplus V_1$, where $V_0$ and $V_1$ are two subspaces respectively given by  $V_0=\lbrace T_{a_0} \rbrace$ and $V_1=\lbrace T_{a_1} \rbrace$. When $\left[V_0,V_1 \right]\subset V_1$, that is to say, when the commutation relations between generators present the following form
\begin{eqnarray}
\left[T_{a_0},T_{b_0}\right]&=& C_{a_0 b_0}^{\;\;\;\;\;c_0} T_{c_0} + C_{a_0b_0}^{\;\;\;\;\;c_1}T_{c_1}, \label{commreduced1}\\ 
\left[T_{a_0},T_{b_1}\right]&=& C_{a_0 b_1}^{\;\;\;\;\;c_1} T_{c_1}, \label{commreduced2}\\ 
\left[T_{a_1},T_{b_1}\right]&=& C_{a_1b_1}^{\;\;\;\;\;c_0}T_{c_0}+C_{a_1b_1}^{\;\;\;\;\;c_1}T_{c_1}, \label{commreduced3}
\end{eqnarray}
one can show that the structure constants $C_{a_0b_0}^{\;\;\;\;\;c_0}$ satisfy the Jacobi identity themselves, and therefore 
\begin{equation}
\left[T_{a_0},T_{b_0}\right]= C_{a_0b_0}^{\;\;\;\;\;c_0}T_{c_0} 
\end{equation}
itself corresponds to a Lie algebra, which is called \textit{reduced algebra} of $\mathfrak{g}$. 
\end{definition}
Let us observe that, in general, a reduced algebra does \textit{not} correspond to a subalgebra.

\subsection{$0_S$-reduction of $S$-expanded algebras}

The $0_S$-reduction of Lie algebras \cite{Izaurieta} involves the extraction of a smaller algebra from an $S$-expanded Lie algebra $\mathfrak{g}_S$, when certain conditions are met. 

In order to give a brief review of $0_S$-reduction, let us consider an abelian semigroup $S$ and the $S$-expanded (super)algebra $\mathfrak{g}_S= S \times \mathfrak{g}$. 

When the semigroup $S$ has a \textit{zero element} $0_S \in S$ (also denoted with the symbol $\lambda_{0_S} \equiv 0_S$ in the literature), this element plays a peculiar role in the $S$-expanded (super)algebra.
In fact, we can split the semigroup $S$ into non-zero elements $\lambda_{i}$, $i=0,...,N$, and a zero element $\lambda_{N+1}=0_S= \lambda_{0_S}$. 
The zero element $\lambda_{0_S}$ is defined as one for which
\begin{equation}
\lambda_{0_S} \lambda_\alpha = \lambda_\alpha \lambda_{0_S} = \lambda_{0_S},
\end{equation}
for each $\lambda_\alpha \in S$.
Under this assumption, we can write $S = \lbrace \lambda_{i}\rbrace \cup \lbrace\lambda_{N+1}=\lambda_{0_S}\rbrace$, with $i = 1, ... ,N$ (the Latin index run only on the non-zero elements of the semigroup). 
Then, the two-selector satisfies the relations
\begin{eqnarray}
K_{i,N+1}^{\;\;\;\;\;\;\;\;\;\; j} &=& K_{N+1,i}^{\;\;\;\;\;\;\;\;\;\; j}=0, \\
K_{i,N+1}^{\;\;\;\;\;\;\;\;\;\; N+1} &=& K_{N+1,i}^{\;\;\;\;\;\;\;\;\;\; N+1}=1, \\
K_{N+1,N+1}^{\;\;\;\;\;\;\;\;\;\;\;\;\;\;\;\;j} &=& 0, \\
K_{N+1,N+1}^{\;\;\;\;\;\;\;\;\;\;\;\;\;\;\;\;N+1} &=& 1,
\end{eqnarray}
which mean, from the multiplication rules point of view,
\begin{eqnarray}
\lambda_{N+1}\lambda_i &=& \lambda_{N+1}, \\
\lambda_{N+1}\lambda_{N+1} &=& \lambda_{N+1}.
\end{eqnarray}
Therefore, for $\mathfrak{g}_S=S \times \mathfrak{g}$ we can write the commutation relations
\begin{eqnarray}
\left[T_{(A,i)},T_{(B,j)}\right]&=&K_{ij}^{\;\;\;k}C_{AB}^{\;\;\;\;C}T_{(C,k)} + K_{ij}^{\;\;\; N+1}C_{AB}^{\;\;\;\;C}T_{(C,N+1)}, \\
\left[T_{(A,N+1)},T_{(B,j)}\right]&=&C_{AB}^{\;\;\;\;C}T_{(C,N+1)}, \\
\left[T_{(A,N+1)},T_{(B,N+1)}\right]&=&C_{AB}^{\;\;\;\;C}T_{(C,N+1)}.
\end{eqnarray}
If we now compare these commutation relations with (\ref{commreduced1}), (\ref{commreduced2}), and (\ref{commreduced3}), we clearly see that
\begin{equation}\label{eqredalg}
\left[T_{(A,i)},T_{(B,j)}\right]= K_{ij}^{\;\;\;k}C_{AB}^{\;\;\;\;C}T_{(C,k)}
\end{equation}
are the commutation relations of a reduced Lie algebra generated by $\lbrace T_{(A,i)} \rbrace$, whose structure constants are given by $K_{ij}^{\;\;\;\;k}C_{AB}^{\;\;\;\;C}$. 

The reduction procedure, in this particular case, is tantamount to impose the condition
\begin{equation}
T_{A,N+1}=\lambda_{0_S} T_A \equiv {0_S} T_A = 0.
\end{equation}
We can notice that, in this case, the reduction abelianizes large sectors of the (super)algebra, and that for each $j$ satisfying $K_{ij}^{\;\;\;N+1}=1$ (that is to say, $\lambda_i \lambda_j= \lambda_{N+1}$), we have 
\begin{equation}
\left[ T_{(A,i)},T_{(B,j)}\right]=0.
\end{equation}
The above considerations led the authors of Ref. \cite{Izaurieta} to the following definition:

\begin{definition}\label{def0Sreduced}
Let $S$ be an abelian semigroup with a zero element $\lambda_{0_S}\equiv 0_S \in S$, and let $\mathfrak{g}_S = S \times \mathfrak{g}$ be an $S$-expanded algebra. Then, the algebra obtained by imposing the condition
\begin{equation}\label{zero}
\lambda_{0_S} T_A \equiv 0_S T_A =0
\end{equation}
on $\mathfrak{g}_S$ (or on a subalgebra of it) is called \textit{$0_S$-reduced algebra} of $\mathfrak{g}_S$ (or of the subalgebra).
\end{definition}

Furthermore, when a $0_S$-reduced algebra presents a structure which is resonant with respect to the structure of the semigroup involved in the $S$-expansion process, the procedure takes the name of \textit{$0_S$-resonant-reduction}.

\subsection{Resonant subalgebras}

Another prescription for getting smaller algebras (subalgebras) from the expanded ones, which read $S \times \mathfrak{g}$, is described in the definition below (see Ref. \cite{Izaurieta}).

\begin{definition}\label{defresonant}
Let $\mathfrak{g}=\bigoplus_{p\in I}V_p$ be a decomposition of $\mathfrak{g}$ into subspaces $V_p$, where $I$ is a set of indices. For each $p, q \in I$ it is always possible to define the subsets $i_{(p,q)} \subset I$, such that 
\begin{equation}\label{decomposition}
\left[V_p,V_q\right]\subset \bigoplus_{r\in i_{(p,q)}} V_r,
\end{equation}
where the subsets $i_{(p,q)}$ store the information on the subspace structure of $\mathfrak{g}$.

Now, let $S=\bigcup_{p\in I} S_p$ be a subset decomposition of the abelian semigroup $S$, such that
\begin{equation}\label{groupdecomposition}
S_p\cdot S_q \subset \bigcap_{r \in i_{(p,q)}} S_r,
\end{equation}
where the product $S_p \cdot S_q$ is defined as
\begin{equation}\label{prod}
S_p \cdot S_q = \lbrace \lambda_\gamma \mid \lambda_\gamma= \lambda_{\alpha_p}\lambda_{\alpha_q}, \; \text{with} \; \lambda_{\alpha_p}\in S_p, \lambda_{\alpha_q}\in S_q \rbrace \subset S.
\end{equation} 
When such subset decomposition $S =\bigcup_{p\in I} S_p$ exists, this decomposition is said to be \textit{in resonance} with the subspace decomposition of $\mathfrak{g}$, $\mathfrak{g}= \bigoplus_{p\in I}V_p$. 
\end{definition}

The resonant subset decomposition is essential in order to systematically extract subalgebras from the $S$-expanded algebra $\mathfrak{g}_S = S \times \mathfrak{g}$, as it was enunciated and proven in Ref. \cite{Izaurieta} with the following theorem:
\begin{theorem}\label{Tres} 
Let $\mathfrak{g} = \bigcup_{p\in I} V_p$ be a subspace decomposition of $\mathfrak{g}$, with a structure described by equation (\ref{decomposition}), and let $S=\bigcup_{p\in I} S_p$ be a resonant subset decomposition of the abelian semigroup $S$, with the structure given in equation (\ref{groupdecomposition}). Define the subspaces of $\mathfrak{g}_S = S \times \mathfrak{g}$ as
\begin{equation}
W_p=S_p \times V_p, \;\;\; p\in I.
\end{equation}
Then, 
\begin{equation}
\mathfrak{g}_R=\bigoplus_{p\in I}W_p
\end{equation}
is a subalgebra of $\mathfrak{g}_S=S \times \mathfrak{g}$, called resonant subalgebra of $\mathfrak{g}_S$.
\end{theorem}

\subsection{In\"on\"u-Wigner contraction}

In what follows, we review the In\"on\"u-Wigner contraction (and, in particular, an instance of its definition, which does fit the definition of standard In\"on\"u-Wigner contraction) and the so-called ``generalized In\"on\"u-Wigner contraction" procedures (see Refs. \cite{IW1, Inonu, Azcarraga13, Gilmorebook, WW1,WW2,Lukierski} for further details). 

\begin{definition}\label{defstandgenIW}
Let $T_i$, $i=1,2,\ldots,n$, be a set of basis vectors for a Lie algebra $\mathfrak{g}$. Let a new set of basis vectors $\tilde{T}_i$, $i=1,2,\ldots,n$, be related to the $T_i$'s by 
\begin{equation}
\tilde{T}_j=U(\varepsilon)_j^{\ i}T_i, \ \ \ U(\varepsilon=1)_j^{\ i}=\delta_j^{\ i} ,\ \ \ det\left[U(\varepsilon=0\right)_j^{\ i}]=0.
\end{equation}
The structure constants of the Lie algebra $\mathfrak{g}$ with respect to the new basis are given by 
\begin{equation}
\left[\tilde{T}_i,\tilde{T}_j\right]=C_{ij}^{\ \ k}(\varepsilon)\tilde{T}_k.
\end{equation}
When the limit 
\begin{equation}
\lim_{\varepsilon\longrightarrow0}C_{ij}^{\ \ k}(\varepsilon)={C'}_{ij}^{\ \ k}
\end{equation}
exists and is well defined, the new structure constants ${C'}_{ij}^{\ \ k}$ characterize a Lie algebra that is not isomorphic to the original one.
This procedure is called \textit{In\"on\"u-Wigner contraction}.
\end{definition}

The above definition fits the standard In\"on\"u-Wigner contraction, which corresponds to a specific choice of the matrix $U(\varepsilon)$, as well as the generalized In\"on\"u-Wigner contraction, which correspond to a special case of Definition \ref{defstandgenIW}, in which the matrix $U(\varepsilon)$ act in a peculiar way, as we will see in the following.

We now discuss a particular instance of Definition \ref{defstandgenIW} which does fit the case of \textit{standard In\"on\"u-Wigner contraction}:
If we now consider the symmetric coset $\mathfrak{g}=\mathfrak{h}\oplus\mathfrak{p}$ of simple Lie algebras, where $\mathfrak{h}$ closes a subalgebra and $\mathfrak{p}$ is a complementary subspace, with commutation relations of the form
\begin{align}
 \left[\mathfrak{h},\mathfrak{h}\right]& \subset\mathfrak{h},\\
 \left[\mathfrak{h},\mathfrak{p}\right]&\subset\mathfrak{p},\\
  \left[\mathfrak{p},\mathfrak{p}\right]&\subset\mathfrak{h},
\end{align}
the standard In\"on\"u-Wigner contraction of $\mathfrak{g}\longrightarrow\mathfrak{g}'$ involves the matrix
\begin{equation}
U(\varepsilon)=
\begin{pmatrix}
I_{dim(\mathfrak{h})} & 0\\
  0 & \varepsilon I_{dim(\mathfrak{p})}
\end{pmatrix} ,
\end{equation}
where $I$ is the identity matrix, and $dim(\mathfrak{h})$ and $dim(\mathfrak{p})$ stand for the dimension of $\mathfrak{h}$ and $\mathfrak{p}$, respectively. We can thus write
\begin{align}
 \left(\begin{array}{c}
  \mathfrak{h}'\\
  \mathfrak{p}'
 \end{array}\right)
=\left(\begin{array}{cc}
  I_{dim(\mathfrak{h})} & 0\\
  0 & \varepsilon I_{dim(\mathfrak{p})}
 \end{array}\right)\left(\begin{array}{c}
             \mathfrak{h}\\
             \mathfrak{p}
            \end{array}\right).
\end{align}
Therefore, after having performed the limit $\varepsilon \longrightarrow 0$, the contracted algebra $\mathfrak{g}'$ becomes a semidirect product $\mathfrak{g}'=\mathfrak{h}'\ltimes\mathfrak{p}'=\mathfrak{h}\ltimes\mathfrak{p}'$,
and we can write the following commutation relations:
\begin{align}
 \left[\mathfrak{h}',\mathfrak{h}'\right]&\subset\mathfrak{h}',\\
 \left[\mathfrak{h}',\mathfrak{p}'\right]&\subset\mathfrak{p}',\\
  \left[\mathfrak{p}',\mathfrak{p}'\right]&=0,
\end{align}
from which we can see that $\mathfrak{p}'$ is now an abelian sector.

If we consider non-symmetric cosets (namely, in the case in which $\left[\mathfrak{p},\mathfrak{p}\right]\subset\mathfrak{h}\oplus \mathfrak{p}$), we can prove that the standard In\"{o}n\"{u}-Wigner contraction still abelianizes this commutator ($\left[\mathfrak{p}',\mathfrak{p}'\right]=0$).

Let us observe that the contracted algebra has the same dimension of the starting one.

The In\"{o}n\"{u}-Wigner contraction has a lot of applications in Physics, among which the relevant cases of the Galilei algebra as an In\"{o}n\"{u}-Wigner contraction of the Poincar\'e algebra, and the Poincar\'e algebra as a contraction of the de Sitter algebra. Both of them involve two universal constants: the velocity of light and the cosmological constant, respectively.

\subsubsection{Generalized In\"{o}n\"{u}-Wigner contraction}

Any diagonal contraction is equivalent to a generalized In\"{o}n\"{u}-Wigner contraction (which can also be referred to as ``Weimar-Woods generalized contraction") \cite{WW1,WW2,Lukierski} with integer parameter powers.

\begin{definition}\label{defgeneralizedIW}
Let $\mathfrak{g}$ be a non-simple algebra which can be written as a sum of $n+1$ subspaces (sets of generators) $V_i$, $i=0,1,\ldots,n $, 
\begin{equation}\label{sumspaces}
 \mathfrak{g}=\bigoplus _{i=0}^{n} V_i = V_0\oplus V_1\oplus\cdots\oplus V_n ,
\end{equation}
such that the following Weimar-Woods conditions \cite{WW1,WW2} are satisfied:
\begin{equation}\label{conditionsWW}
\left[V_p,V_q\right]\subset\bigoplus_{s\leq p+q}V_s,\ \ p,q=0,1,\ldots,n .
\end{equation}
The conditions in (\ref{conditionsWW}) imply that $V_0$ is a subalgebra of $\mathfrak{g}$. The contraction of $\mathfrak{g}$ can be obtained after having performed a proper rescaling on the generators of each subspace, and once a singular limit for the contraction parameter have been considered, namely by considering
\begin{equation}\label{matrUspec}
{V'}_i=\varepsilon^{a_i}V_i, \;\;\; i=0,1,\ldots,n ,
\end{equation}
with the choice of the powers $a$'s providing a finite limit of the contracted algebra if $\varepsilon\longrightarrow0$. This procedure is called \textit{generalized In\"on\"u-Wigner contraction}.
\end{definition}

Let us observe that Definition \ref{defgeneralizedIW} corresponds to a special case of Definition \ref{defstandgenIW}, in which the matrix $U(\varepsilon)$ acts as in (\ref{matrUspec}).
In fact, the generalized In\"{o}n\"{u}-Wigner contractions are produced by diagonal matrices of the form $U(\varepsilon)_{jj}= \varepsilon^{n_j},\ n_j\in\mathbb{Z}$.
Thus, a contraction $\mathfrak{g}\xrightarrow{U(\varepsilon)_{ij}}\mathfrak{g}'$ is called \textit{generalized In\"on\"u-Wigner contraction} \cite{WW2} if the matrix $U(\varepsilon)$ has the form
\begin{equation}
 U(\varepsilon)_{ij}=\delta_{ij}\varepsilon^{n_j},\ \ n_j\in \mathbb{R},\ \ \varepsilon>0,\ \ i,j=1,2,\ldots ,N,
\end{equation}
with respect to a basis of generators $\lbrace{T_{a_1}, T_{a_2},\ldots,T_{a_N}\rbrace}$.

\subsection{Ideal subalgebra} \label{ideal1}

In the following, we extend the concept of normal subgroup of a group to the concept of ideal subalgebra of an algebra (see Ref. \cite{Jacobson} for further details). 

\begin{definition}\label{defideal}
Let us consider a Lie group $G$ and a normal subgroup $H$ of $G$. Let the Lie algebra $\mathcal{A}$ be the algebra associated with the Lie group $G$ (by exponentiation), and let the subalgebra $\mathcal{I}$ of $\mathcal{A}$ be the subalgebra associated with the normal subgroup $H$. Then, $\mathcal{I}$ is an \textit{ideal} (\textit{ideal subalgebra}) of $\mathcal{A}$, and we can write
\begin{equation}
\left[\mathcal{A},\mathcal{I}\right]\subset\mathcal{I}.
\end{equation}
\end{definition}

Let us now consider the homomorphism $\varphi: G\longrightarrow G/H$, where $G$ is a Lie group and $H$ a normal subgroup of $G$. Let $\mathcal{A}$ be the Lie algebra associated with the Lie group $G$. \\ By definition, $H=ker \varphi$, and $G/H$ is a Lie group. Let $\hat{\mathcal{A}}$ be the Lie algebra associated with the Lie group $G/H$. Then, $\varphi$ induces a homomorphism $\hat{\varphi}: \mathcal{A}\longrightarrow\hat{\mathcal{A}}$ between algebras, such that, if $a\in\mathcal{A}$, then $\varphi(e^a)=e^{\hat{\varphi}(a)}$. Since, from Definition \ref{defideal}, $\forall a \in \mathcal{I}$ we have $e^a\in H$, this now implies $\varphi(e^a)=e' \in G/H \ \Longleftrightarrow \ \hat{\varphi}(a)=0$.
Consequently, $\hat{\varphi}(\mathcal{I})=\left\{0\right\}$. 

The algebra $\hat{\mathcal{A}}$ is isomorphic to the coset space $\mathcal{A}\ominus\mathcal{I}$, whose elements are the equivalence classes $\left[a\right]=\left\{a'\in\mathcal{A}:a'-a\in\mathcal{I}\right\}$.

Let us now consider $\rho$: $G \longrightarrow Aut(V)$, where $V$ is a vector space. This defines a representation of $G/H$ provided on the vector space $V$. It has a trivial action on each vector $\Psi \in V$:
\begin{equation}
 \rho(H) \Psi=\Psi . 
\end{equation}
Then, by defining, $\forall \left[g\right]\in G/H$,
\begin{equation}
 \rho(\left[g\right])=\rho(g),
\end{equation}
since
\begin{align}
 \rho(g\cdot h)\Psi=&\rho(g)\rho(h)\Psi\\
 =&\rho(g)\Psi,\ \ \forall \Psi \in V,
\end{align}
we can say that $\rho(G/H)=\rho(G)$
is the property we need in order to require 
\begin{equation}
\rho(\mathcal{I})=0 .
\end{equation}
In this way, $\rho$ provides a representation of $\mathcal{A}\ominus\mathcal{I}$ such that:
\begin{equation}
\rho(\mathcal{A})=\rho(\mathcal{A}\ominus\mathcal{I}).
\end{equation}
Thus, if we now write
\begin{equation}
\mathcal{A}\ominus\mathcal{I} = \mathcal{A}_0,
\end{equation}
where $\mathcal{A}_0$ is a coset space, we can say that $\rho(\mathcal{A}_0)$ is homomorphic to $\rho(\mathcal{A})$, and we can finally write
\begin{equation}
 \rho(\mathcal{A})\Psi=\rho(\mathcal{A}_0)\Psi,\ \ \forall\Psi \in V.
\end{equation}

\section{In\"on\"u-Wigner contraction and $S$-expansion} \label{Generalized}

In the following, we show that a standard In\"on\"u-Wigner contraction can be reproduced with a (finite) $S$-expansion procedure (with resonance and $0_S$-reduction). Then, we review the way in which the authors of \cite{Izaurieta} used the information on the subspace structure of the original Lie algebra to extract reduced algebras from resonant subalgebras of $S$-expanded algebras, being able to reproduce, in this way, the generalized In\"on\"u-Wigner contraction.

After doing this, we proceed to the core of the present paper, that is we describe our prescription for infinite $S$-expansion (involving an infinite abelian semigroup $S^{(\infty)}$) with ideal subtraction, and we explain how to reproduce a generalized In\"on\"u-Wigner contraction in this context. This method is a new prescription for $S$-expansion and also an alternative way of seeing the (generalized) In\"on\"u-Wigner contraction.

\subsection{Finite $S$-expansion and In\"{o}n\"{u}-Wigner contraction}

We first of all show the way in which a standard In\"on\"u-Wigner contraction can be reproduced with a (finite) $S$-expansion procedure. Then, we review, following \cite{Izaurieta}, the procedure to extract reduced algebras from resonant subalgebras of $S$-expanded algebras, allowing the $S$-expansion method to reproduce the case of generalized In\"on\"u-Wigner contraction.

\subsubsection{Standard In\"{o}n\"{u}-Wigner contraction as finite $S$-expansion with resonance and $0_S$-reduction}

Let us consider a coset $\mathfrak{g}$ with a subalgebra $\mathfrak{h}$ and a complementary subspace $\mathfrak{p}$, namely $\mathfrak{g}=\mathfrak{h}\oplus \mathfrak{p}$, with commutation relations of the form
\begin{align}
 \left[\mathfrak{h},\mathfrak{h}\right]\subset&\mathfrak{h},\nonumber\\
 \left[\mathfrak{h},\mathfrak{p}\right]\subset&\mathfrak{p},\label{Symmetric}\\
 \left[\mathfrak{p},\mathfrak{p}\right]\subset&\mathfrak{h}\oplus\mathfrak{p}\nonumber.
\end{align}
The In\"on\"u-Wigner contraction of $\mathfrak{g} \longrightarrow \mathfrak{g}'$ involves the following transformation:
\begin{equation}
 \begin{pmatrix}
  \mathfrak{h}'\\
  \mathfrak{p}'
 \end{pmatrix}=\begin{pmatrix}
  I_{dim(\mathfrak{h})} & 0\\
 0 & \varepsilon I_{dim(\mathfrak{p})}\\
\end{pmatrix}\begin{pmatrix}
\mathfrak{h}\\
\mathfrak{p}
\end{pmatrix},
\end{equation}
where $I$ is the identity matrix, and $dim(\mathfrak{h})$ and $dim(\mathfrak{p})$ stand for the dimension of $\mathfrak{h}$ and $\mathfrak{p}$, respectively; $\varepsilon$ is the contraction parameter. Then, the commutation relations of $\mathfrak{g}'=\mathfrak{h}'\ltimes \mathfrak{p}'=\mathfrak{h}\ltimes \mathfrak{p}'$ are well defined for all values of $\varepsilon$, including the singular limit $\varepsilon \longrightarrow0$:
\begin{align}
  \left[\mathfrak{h},\mathfrak{h}\right]\subset&\mathfrak{h}, & & &\left[\mathfrak{h}',\mathfrak{h}'\right]\subset\mathfrak{h}',\nonumber\\
 \left[\mathfrak{h},\mathfrak{p}\right]\subset&\mathfrak{p},&\xrightarrow{\mathfrak{h}'=\mathfrak{h},\;\mathfrak{p}'=\varepsilon \mathfrak{p},\; \varepsilon \longrightarrow0} & &\left[\mathfrak{h}',\mathfrak{p}'\right]\subset\mathfrak{p}',\label{IWstandard}\\
 \left[\mathfrak{p},\mathfrak{p}\right]\subset&\mathfrak{h}\oplus\mathfrak{p}, &  & & \left[\mathfrak{p}',\mathfrak{p}'\right]\ =0.\nonumber
\end{align}

This In\"on\"u-Wigner contraction (standard In\"on\"u-Wigner contraction) can be reproduced with a finite $S$-expansion procedure involving $0_S$-resonant-reduction and the semigroup $S_E^{(1)}=\left\{\lambda_0,\lambda_1,\lambda_2\right\}$ (where the zero element of the semigroup is $\lambda_{0_S}=\lambda_2$), described by the multiplication table
\begin{equation}
 \begin{array}{c|ccc}
  & \lambda_0 & \lambda_1 & \lambda_2\\
 \hline
 \lambda_0 & \lambda_0 & \lambda_1&\lambda_2\\
 \lambda_1 & \lambda_1 & \lambda_2 &\lambda_2\\
 \lambda_2 & \lambda_2 & \lambda_2 &\lambda_2
 \end{array}
\end{equation}
This can be done (see Ref. \cite{Izaurieta}) by considering a proper partition over the subspaces of the starting Lie algebra and by multiplying its generators by the elements of the semigroup $S_E^{(1)}$.
Then, the resonant-reduced, S-expanded Lie algebra $\mathfrak{g}_{S_{RR}}=\mathfrak{h}'\oplus \mathfrak{p}'$, where $\mathfrak{h}' = \lambda_0 \mathfrak{h}$ and $\mathfrak{p}'= \lambda_1 \mathfrak{p}$, satisfies the following commutation relations:
\begin{align}
  \left[\lambda_0\mathfrak{h},\lambda_0\mathfrak{h}\right]\subset&\lambda_0\mathfrak{h}, \\
 \left[\lambda_0\mathfrak{h},\lambda_1\mathfrak{p}\right]\subset&\lambda_1\mathfrak{p},\\
 \left[\lambda_1\mathfrak{p},\lambda_1\mathfrak{p}\right]=&0.
\end{align}
Here we can see that the role of the zero element $\lambda_2=\lambda_{0_S}$ is to turn to zero each multiplicand (see equation (\ref{zero}) in Section \ref{Review}).
We can easily see that by rewriting $\lambda_0\mathfrak{h}$ and $\lambda_1 \mathfrak{p}$ in terms of $\mathfrak{h}'$ and $\mathfrak{p}'$, we arrive to equation (\ref{IWstandard}). 
Thus, we can conclude that the In\"{o}n\"{u}-Wigner contraction can be seen as an $S$-expansion (involving $0_S$-resonant-reduction) performed with the semigroup $S_E^{(1)}$.

\subsubsection{Reduction of resonant subalgebras and generalized In\"{o}n\"{u}-Wigner contraction}

The authors of \cite{Izaurieta} proved that one can extract reduced algebras from the resonant subalgebra of a (finite) $S$-expanded algebra $S \times \mathfrak{g}$, and they showed that, by following this path, the generalized In\"{o}n\"{u}-Wigner contraction fits within their scheme.

In particular, Theorem VII.1 of their work provides necessary conditions under which a reduced algebra can be extracted from a resonant subalgebra. It reads as follow:

\begin{theorem}\label{teoiza}
Let $\mathfrak{g}_R = \bigoplus_{p \in I} S_p \times V_p$ be a resonant subalgebra of $\mathfrak{g}_S=S \times \mathfrak{g}$. Let $S_p = \hat{S}_p \cup \check{S}_p$ be a partition of the subset $S_p \subset S$ such that
\begin{equation}\label{intzero}
\hat{S}_p \cap \check{S}_p = \emptyset ,
\end{equation}
\begin{equation}\label{ressemigroup}
\check{S}_p \cdot \hat{S}_q \subset \bigcap_{r \in i_{(p,q)}}\hat{S}_r.
\end{equation}
Conditions (\ref{intzero}) and (\ref{ressemigroup}) induce the decomposition $\mathfrak{g}_R = \check{\mathfrak{g}}_R \oplus \hat{\mathfrak{g}}_R$ on the resonant subalgebra, where
\begin{equation}
\check{\mathfrak{g}}_R = \bigoplus_{p \in I} \check{S}_p \times V_p ,
\end{equation}
\begin{equation}
\hat{\mathfrak{g}}_R = \bigoplus_{p \in I}\hat{S}_p \times V_p .
\end{equation}
When the conditions (\ref{intzero}) and (\ref{ressemigroup}) hold, then
\begin{equation}
\left[ \check{\mathfrak{g}}_R, \hat{\mathfrak{g}}_R \right] \subset \hat{\mathfrak{g}}_R,
\end{equation}
ant therefore $\vert \check{\mathfrak{g}}_R \vert$ corresponds to a reduced algebra of $\mathfrak{g}_R$.
\end{theorem}

Indeed, $\mathfrak{g}_R$ fits the definition of reduced algebra (\ref{defreduced}).
Using the structure constants for the resonant subalgebra, it is then possible to find the structure constants for the reduced algebra $\vert \check{\mathfrak{g}}_R \vert$, as shown in \cite{Izaurieta}.

Let us observe that when every $S_p \subset S$ of a resonant subalgebra includes the zero element $\lambda_{0_S}$, the choice $\hat{S}_p = \lbrace \lambda_{0_S} \rbrace$ automatically satisfies the conditions (\ref{intzero}) and (\ref{ressemigroup}). As a consequence of this, the $0_S$-reduction previously reviewed can be regarded as a particular case of Theorem \ref{teoiza}.

Theorem \ref{teoiza}, namely, Theorem VII.1 of \cite{Izaurieta}, has been then used by the author of the same paper to recover Theorem 3 of Ref. \cite{Azcarraga1} in the context of Lie algebras contractions.

Considering the semigroup $S^{(N)}_E = \lbrace \lambda_\alpha , \alpha = 0, \ldots, N+1 \rbrace$, provided with the multiplication rule
\begin{equation}
\lambda_\alpha \lambda_ \beta = \lambda_{H_{N+1}(\alpha + \beta)},
\end{equation}
where $H_{N+1}$ is defined as the function
\begin{equation}
H_n (x) =  \left\{ \begin{aligned} &
x , \;\;\;\;\; \text{when} \;  x < n,
\\ & n , \;\;\;\;\; \text{when} \; x\geq n, \end{aligned} 
\right.
\end{equation}
and where $\lambda_{N+1}$ is the zero element in $S^{(N)}_E$, namely $\lambda_{N+1}= \lambda_{0_S}$, the authors of \cite{Izaurieta} performed $S$-expansions using $S=S_E$ and found resonant partitions for $S_E$, constructing resonant subalgebras $\mathfrak{g}_R$'s and then applying a $0_S$-reduction to the resonant subalgebras, in different cases. Their aim, in this context, was to recover some results presented for algebra expansions in \cite{Azcarraga1} within the $S$-expansion approach.

In particular, they analyzed the case in which a Lie algebra $\mathfrak{g}$ fulfills the Weimar-Woods conditions \cite{WW1, WW2}, namely when one can write the subspace decomposition $\mathfrak{g}= \bigoplus_{p=0}^n V_p$ of the Lie algebra $\mathfrak{g}$ for which the following Weimar-Woods conditions hold:
\begin{equation}
\left[V_p , V_q \right] \subset \bigoplus_{p=0}^{H_n (p+q)} V_r .
\end{equation}
In this context, the authors of \cite{Izaurieta} considered a subset decomposition of $S_E$
\begin{equation}\label{subsetdec}
S_E = \bigcup_{p=0}^n S_p ,
\end{equation}
where the subsets $S_p \subset S_E$ were defined by
\begin{equation}
S_p = \lbrace \lambda_{\alpha_p}, \; \alpha_p = p , \ldots, N+1 \rbrace ,
\end{equation}
with $N+1 \geq n$. The subset decomposition (\ref{subsetdec}) is a resonant one under the semigroup product (\ref{prod}), since it satisfies
\begin{equation}
S_p \cdot S_q = S_{H_n(p+q)} \subset \bigcap_{r=0}^{H_n(p+q)} S_r.
\end{equation}
Thus, according to Theorem \ref{Tres} (which corresponds to Theorem IV.2 of \cite{Izaurieta}), the direct sum
\begin{equation}
\mathfrak{g}_R = \bigoplus_ {p=0}^n W_p,
\end{equation}
with
\begin{equation}
W_p = S_p \times V_p ,
\end{equation}
is a resonant subalgebra of $\mathfrak{g}$.
Then, they considered the following $S_p$ partition, which satisfies (\ref{intzero}):
\begin{equation}\label{part1}
\check{S}_p = \lbrace \lambda_{\alpha_p}, \; \alpha_p = p , \ldots , N_p \rbrace , 
\end{equation}
\begin{equation}\label{part2}
\hat{S}_p = \lbrace \lambda_{\alpha_p}, \; \alpha_p = N_p + 1, \ldots , N+1 \rbrace ,
\end{equation}
and they proved that the reduction condition (\ref{ressemigroup}) on (\ref{part1}) and (\ref{part2}) is equivalent to the following requirement on the $N_p$'s:
\begin{equation}
N_{p+1} = \left\{ \begin{aligned} &
N_p , \;\;  \text{or} 
\\ & H_{N+1}(N_p+1) .
\end{aligned} 
\right.
\end{equation}
This condition is exactly the one obtained in Theorem 3 of Ref. \cite{Azcarraga1} (requiring that the expansion in the Maurer-Cartan forms closes). In the $S$-expansion context, the case
\begin{equation}
N_{p+1} = N_p = N+1
\end{equation}
for each $p$ corresponds to the resonant subalgebra, and the case
\begin{equation}
N_{p+1} = N_p = N
\end{equation}
to its $0_S$-reduction \cite{Izaurieta}. 
The generalized In\"{o}n\"{u}-Wigner contraction corresponds to the case 
\begin{equation}
N_p = p
\end{equation}
(see Ref. \cite{Azcarraga1} for details). This, as stated in \cite{Izaurieta}, means that the generalized In\"{o}n\"{u}-Wigner contraction does not correspond to a resonant subalgebra, but to its reduction. This is an important point, because, the author of \cite{Izaurieta} have been able to define non-trace invariant tensors for resonant subalgebras and $0_S$-reduced algebras, but \textit{not} for general reduced algebras.

As we will see in the following, the prescription for infinite $S$-expansion with ideal subtraction we develop in the present paper leads to reduced algebras. In particular, the subtraction of the infinite ideal subalgebra can be viewed as a $0_S$-reduction (since it reproduces the same result) involving an infinite number of semigroup elements that, together with the generators associated, play the role of ``generating zeros". In this context, as we will discuss, we are able to write the invariant tensors of the target algebras in terms of those on the original ones.

We can thus conclude that, in the context of finite $S$-expansion, the path for reproducing a generalized In\"{o}n\"{u}-Wigner contraction consists in extracting reduced algebras from the resonant subalgebra of an $S$-expanded algebra.

\subsection{Infinite $S$-expansion with ideal subtraction}\label{infexp}

As we have previously discussed, in Ref. \cite{Izaurieta} generalized In\"{o}n\"{u}-Wigner contractions have been realized as reductions of $S$-expanded algebras with a finite semigroup.

In the following, we will discuss a new prescription for $S$-expansion, involving an \textit{infinite} abelian semigroup and the subsequent subtraction of an infinite ideal. This scenario also offers an alternative view on the generalized In\"{o}n\"{u}-Wigner contraction procedure. As we have already mentioned in the introduction, the subtraction of the infinite ideal in this context is crucial, since it allows to obtain Lie (super)algebras with a finite number of generators, after having infinitely expanded the original Lie (super)algebras. 

Before proceeding to the development of our method, let us spend a few words on the general idea. Thus, let us remind that the generalized In\"{o}n\"{u}-Wigner contraction (see Ref. \cite{Lukierski}) can be performed when we consider a non-simple algebra $\mathfrak{g}$ decomposed into $n+1$ subspaces $V_i$ of generators $T_a^{(i)}$ ($i=0,1,\ldots , n$), namely
\begin{equation}\label{KM1}
 \mathfrak{g}= V_0 \oplus V_1 \oplus \ldots \oplus V_n  ,
\end{equation}
where the conditions (\ref{conditionsWW}) (namely the Weimar-Woods conditions \cite{WW1, WW2}) are satisfied. 

The generalized In\"{o}n\"{u}-Wigner contraction of (\ref{KM1}) is obtained by properly rescaling each generator $T^{(i)}_a$ by a power of the contraction parameter $\varepsilon$, namely
\begin{equation}
T_a^{(i)} \in V_i \longrightarrow \varepsilon^{a_i} T_a^{(i)} \in V'_i,
\end{equation}
where the choice of $a_i$ provides finite limits of the contracted algebra when $\varepsilon\longrightarrow0$. 

With this in mind, in our method we first of all consider an \textit{infinite} $S$-expansion procedure, namely an $S$-expansion which involves an \textit{infinite abelian semigroup} $S^{(\infty)}$, which will be defined in the following.

Then, we will show how to confer the role of ``generating zeros'' to a particular infinite set of generators, associating them with different elements of the infinite abelian semigroup $S^{(\infty)}$. 
This will need to pass through \textit{infinite resonant subalgebras}, as we will discuss in a while. We will then be able to subtract an infinite ideal subalgebra from an infinite resonant subalgebra of the infinitely $S$-expanded algebra.
We will explicitly show that the ideal subtraction corresponds to a reduction.
This procedure will also be able to reproduce the same result which would have been obtained by having performed a generalized In\"{o}n\"{u}-Wigner contraction. 

We observe that the algebra we end up after the ideal subtraction, in general, is \textit{not} a subalgebra of the starting algebra. It is, instead, a reduced algebra. In particular, the ideal subtraction reproduces the same result of a $0_S$-reduction.

In this context, we are taking into account the following operation on a given algebra $\mathcal{A}$:
\begin{align}\label{idealalgebra}
\mathcal{A}\ominus\mathcal{I}=\mathcal{A}_0,
\end{align}
where $\mathcal{A}_0$ generates a coset space, and where $\mathcal{I}$ is an ideal (ideal subalgebra) of $\mathcal{A}$, namely a subalgebra of $\mathcal{A}$ that satisfies the property $\left[\mathcal{A},\mathcal{I}\right]\subset\mathcal{I}$ (see Subsection \ref{ideal1} for further details).

We now explain in detail our method involving an infinite $S$-expansion with subsequent ideal subtraction.

\subsubsection{General formulation of the method of infinite $S$-expansion with ideal subtraction}

As said in \cite{Caroca}, if, in the $S$-expansion procedure, the finite semigroup is generalized to the case of an
infinite semigroup, then the $S$-expanded algebra will be an \textit{infinite-dimensional algebra}.
One can this see from the fact that $T_{(A,\alpha)} = \lambda_\alpha T_A$ constitutes a base for the $S$-expanded
algebra and from the fact that $\alpha$ takes now the values in an infinite set.

We can thus generate an infinitely $S$-expanded algebra as a loop-like Lie algebra (see Ref. \cite{Caroca}), where the semigroup elements can be represented by the set $(\mathbb{N},+)$, \footnote{The loop algebra \cite{Caroca} was constructed by considering the semigroup $(\mathbb{Z},+)$ (which is an abelian group with the sum operation). In our work, we restrict to $(\mathbb{N},+)$ and we leave the extension to $(\mathbb{Z},+)$ to the future.} that presents the same multiplication rules (extended to an infinite set) of the general semigroup $S^{(N)}_E = \lbrace{ \lambda_\alpha\rbrace}_{\alpha=0}^{N+1}$, namely $\lambda_\alpha \lambda_\beta = \lambda_{\alpha+\beta}$ if $\alpha+\beta \leq N+1$, and $\lambda_\alpha \lambda_\beta = \lambda_{N+1}$ if $\alpha+\beta > N+1$. \footnote{Different semigroups of the type $S^{(N)}_E$ have been used and discussed in several works on $S$-expanded algebras (see Refs. \cite{Izaurieta,CPRS1,Concha1,Diego,Concha:2016hbt,Pat2,CR2,Concha2}).}

\begin{definition}\label{definftysemigr}
Let $\left\{\lambda_\alpha\right\}_{\alpha=0}^{\infty}=\left\{\lambda_0,\lambda_1,\lambda_2, \ldots, \lambda_\infty \right\}$ be an infinite discrete set of elements. \footnote{We have used the notation $\lambda_\infty$ just for denoting the fact that we are considering an infinite set, in the sense that it contains an infinite number of elements.} Then, the infinite set $\left\{\lambda_\alpha\right\}_{\alpha=0}^{\infty}$ satisfying commutation rules like the ones of the set $(\mathbb{N},+)$ (that is, of $S^{(N)}_E$), namely
\begin{equation}\label{prodsemigrinfinito}
\lambda_\alpha \lambda_\beta = \lambda_{\alpha + \beta},
\end{equation}
where 
\begin{equation}\label{infelementprod}
\lambda_\alpha \lambda_{\infty} = \lambda_\infty , \quad \forall \lambda_\alpha \in \left\{\lambda_\alpha\right\}_{\alpha=0}^{\infty}, \quad \quad \lambda_\infty \lambda_\infty = \lambda_\infty ,
\end{equation}
is an infinite abelian semigoup.
\end{definition}

Hereafter, we will denote such an infinite abelian semigroup by $S^{(\infty)} = \left\{\lambda_\alpha\right\}_{\alpha=0}^{\infty}$.
Let us notice that since the multiplication rules in (\ref{infelementprod}) hold, the element $\lambda_\infty \in S^{(\infty)}$ can be regarded as an ``ideal element" of the infinite semigroup $S^{(\infty)}$.

Now, let $\mathfrak{g} = \bigoplus_{p \in I} V_p$ be a subspace decomposition of $\mathfrak{g}$.

With the above assumptions, we perform an infinite $S$-expansion on $\mathfrak{g}$ using the semigroup $S^{(\infty)}$, and the infinite $S$-expanded algebra can be rewritten as
\begin{align}
\mathfrak{g}_S^\infty&=\left\{\lambda_\alpha\right\}^\infty_{\alpha=0}\times\mathfrak{g}= \nonumber\\
&=\left\{\lambda_\alpha\right\}^\infty_{\alpha=0}\times \left[ \bigoplus_{p \in I} V_p \right] \label{suma}.
\end{align}

Let us observe that the Jacobi identity is fulfilled for the infinite $S$-expanded algebra, since the starting algebra satisfies the Jacobi identity and the semigroup $S^{(\infty)}$ is abelian (and associative by definition): These are the requirements that the starting algebra and the semigroup involved in the procedure must satisfy so that the $S$-expanded algebra satisfies the Jacobi identity when performing an $S$-expansion process (see Ref. \cite{Izaurieta} for further details).

At this point, one can split the infinite semigroup in subsets in such a way to be able to properly extract a resonant subalgebra from the infinitely $S$-expanded one. It will be then possible, according to the procedure described in \cite{Izaurieta}, to define partitions on these subsets such that one can isolate ad ideal structure from the resonant sublagebra of the infinitely $S$-expanded algebra, reproducing a reduction and ending up, in this way, with a finite algebra. This procedure, as we will see, is also able to reproduce a generalized In\"{o}n\"{u}-Wigner contraction. 

We will now describe the general development of the method, which is based on the following steps:
\begin{enumerate}
\item Properly define subsets $S_p$ of $S^{\infty}$ such that they satisfy the resonant condition (\ref{groupdecomposition}), in order to be able to extract a resonant subalgebra from $\mathfrak{g}^\infty_S$ (by using Theorem \ref{Tres}).
\item Explicitly show that the ideal subtraction we want to apply satisfies the requirements of a reduction (in the sense of Definition \ref{defreduced}). To this aim, one must define a $S_p$ partition such that the conditions (\ref{intzero}) and (\ref{ressemigroup}) are fulfilled, in order to extract a reduced algebra from the resonant subalgebra of the infinitely $S$-expanded one. This reduced algebra will be our target, while the (infinite) ideal subalgebra will be taken apart. 
\end{enumerate}

Depending on the subspace decomposition structure of the original algebra $\mathfrak{g}$, the subset decomposition of $S^{(\infty)}$ in subsets $S_p$ will assume different forms in order to satisfy the resonance condition, allowing the extraction of a resonant subalgebra. We will carry on the discussion in general in the following, in order to expose our method, while in Section \ref{Examples} we will study some examples with different subspace decomposition structures, case by case. 

To proceed with the extraction of the infinite resonant subalgebra, according with the review we have presented in Section \ref{Review} and to the approach presented in \cite{Izaurieta}, we must define a resonant subset decomposition, under the product (\ref{prod}), of the infinite semigroup $S^{(\infty)}$
\begin{equation}\label{subsetdecinf1}
S^{(\infty)} = \bigcup_{p \in I} S_p,
\end{equation}
namely a decomposition such that (\ref{groupdecomposition}) is fulfilled, where the $S_p$'s are infinite subsets.

Once the resonant subset decomposition has been found, the direct sum
\begin{equation}
\mathfrak{g}^\infty_R = \bigoplus_{p \in I} W_p ,
\end{equation}
with 
\begin{equation}
W_p = S_p \times V_p , \quad p \in I ,
\end{equation}
is a \textit{resonant subalgebra} of $\mathfrak{g}^\infty_S$ (where we have used Theorem \ref{Tres}). 

In particular, $\mathfrak{g}^\infty_R = \bigoplus_{p\in I} W_p = \bigoplus_{p \in I} S_p \times V_p$ is the direct sum of a finite number of infinite subspaces $W_p$, which are infinite due to the fact that the subsets $S_p$'s contains an infinite amount of semigroup elements.

Now, according with the procedure described in \cite{Izaurieta}, we can develop the following theorem:
\begin{theorem}\label{TeoRedIdeal}
Let $\mathfrak{g}$ be a Lie (super)algebra and let $\mathfrak{g}^\infty_S = S^{(\infty)}\times \mathfrak{g}$ be the infinite $S$-expanded (super)algebra obtained using the infinite abelian semigroup $S^{(\infty)} = \lbrace \lambda_\alpha , \; \alpha = 0 , \ldots , \infty \rbrace$. 

Let $\mathfrak{g}^\infty_R$ be an infinite resonant subalgebra of $\mathfrak{g}^\infty_S$ and let $\mathcal{I}$ be an infinite ideal subalgebra of $\mathfrak{g}^\infty_R$.
Then, the (super)algebra
\begin{equation}
\check{\mathfrak{g}}_R = \mathfrak{g}^\infty_R \ominus \mathcal{I}
\end{equation}
is a reduced (super)algebra.
\end{theorem}

\begin{proof}
After having performed the infinite $S$-expansion on $\mathfrak{g}$, obtaining the infinite $S$-expanded (super)algebra $\mathfrak{g}^\infty_S = S^{(\infty)}\times \mathfrak{g}$, and after having extracted a resonant subalgebra $\mathfrak{g}^\infty_R $ from $\mathfrak{g}^\infty_S$ as described above, we can write an $S_p$ partition $S_p = \hat{S}_p \cup \check{S}_p$, where the $\check{S}_p$'s are finite subsets, while the $\hat{S}_p$'s are infinite ones, satisfying
the conditions (\ref{intzero}) and (\ref{ressemigroup}).
Once such a partition has been found, it induces, according with Theorem \ref{teoiza}, the following decomposition on the resonant subalgebra $\mathfrak{g}^\infty_R$:
\begin{equation}
\mathfrak{g}^\infty_R = \check{\mathfrak{g}}_R \oplus \hat{\mathfrak{g}}^\infty_R ,
\end{equation}
where
\begin{equation}
\check{\mathfrak{g}}_R = \bigoplus_{p\in I} \check{S}_p \times V_p , 
\end{equation}
\begin{equation}
\hat{\mathfrak{g}}^\infty_R = \bigoplus_{p \in I} \hat{S}_p \times V_p.
\end{equation}
Let us observe that $\check{\mathfrak{g}}_R$ is \textit{finite}, since it is the direct sum of products between finite subsets and finite subspaces, while $\hat{\mathfrak{g}}^\infty_R$ is \textit{infinite}, due to the fact that the $\hat{S}_p$'s are infinite subsets.

Then, applying Theorem \ref{teoiza}, we have
\begin{equation}\label{reducedcond1}
\left[\check{\mathfrak{g}}_R , \hat{\mathfrak{g}}^\infty_R\right] \subset  \hat{\mathfrak{g}}^\infty_R ,
\end{equation}
and, therefore, $\vert \check{\mathfrak{g}}_R \vert$ correspond to a reduced (super)algebra of $\mathfrak{g}^\infty_S$.

Furthermore, in the case in which we have
\begin{equation}
\left[\hat{\mathfrak{g}}^\infty_R , \hat{\mathfrak{g}}^\infty_R \right] \subset \hat{\mathfrak{g}}^\infty_R ,
\end{equation}
that is to say, when $\hat{\mathfrak{g}}^\infty_R$ is an infinite subalgebra of $\mathfrak{g}^\infty_R$ (and, consequently, an infinite subalgebra of $\mathfrak{g}^\infty_S$), $\hat{\mathfrak{g}}^\infty_R$ is, in particular, an \textit{infinite ideal subalgebra} (due to the fact that it also satisfies (\ref{reducedcond1})), in the sense described in Section \ref{Review}. 
This allows us to write
\begin{equation}
\check{\mathfrak{g}}_R = \mathfrak{g}^\infty_R \ominus \mathcal{I},
\end{equation}
where we have denoted by $\mathcal{I}$ the infinite ideal subalgebra, $\mathcal{I}\equiv \hat{\mathfrak{g}}^\infty_R$, and where $\check{\mathfrak{g}}_R$ corresponds to the reduced algebra one end up with at the end of the procedure.

We have thus demonstrated that the subtraction of an infinite ideal subalgebra from an infinite resonant subalgebra of an infinite $S$-expanded (super)algebra (reached by using the infinite semigroup $S^{(\infty)}$) corresponds to a \textit{reduction}, leading to a \textit{reduced algebra}.

\end{proof}

The ideal subtraction, which is crucial since it allows to end up with a Lie algebra with a finite number of generators, thus satisfies the requirements of a reduction, in the sense given in Section \ref{Review} in the context of reduced algebras. As we have already mentioned, a reduced algebra, in general, does \textit{not} correspond to a subalgebra.

In particular, \textit{the ideal subtraction can be viewed as a (generalization of the) $0_S$-reduction}, in the sense that all the elements of the infinite ideal are mapped to zero after the ideal subtraction; This has the same effect given by the zero element $\lambda_{0_S}$ of a semigroup (see Section \ref{Review}), namely
\begin{equation}
\lambda_{0_S} T_A = 0 .
\end{equation}
This is what we meant when we said that we are conferring the role of ``generating zeros'' to a particular infinite set of generators, that is the ones belonging to the infinite ideal. Thus, the reduced algebra $\check{\mathfrak{g}}_R$ can be viewed, in this sense, as a $0_S$-reduced algebra.

Using the structure constants that one can write for the resonant subalgebra, it is then possible to find the
structure constants for the ($0_S$)-reduced algebra, as it was presented in \cite{Izaurieta}.

As we will see in Section \ref{Examples}, the method we have described above is reliable and can be applied to algebras already considered in the literature in the context of expansion and contraction, reproducing the same results shown in the literature.
The main point consists in choosing properly the resonant partition of the semigroup $S^{(\infty)}$ and then the $S_p$ partition $S_p = \check{S}_p \cup \hat{S}_p$ (namely, in selecting in a proper way the subsets $\check{S}_p$ and $\hat{S}_p$), in order to be able to extract an infinite ideal subalgebra $\mathcal{I}$ from the infinite resonant subalgebra $\mathfrak{g}^\infty_R$.

\subsubsection{Formulation of the method for reproducing a generalized In\"{o}n\"{u}-Wigner contraction}\label{GENSE}

Let us now apply our method to the case in which the original Lie algebra $\mathfrak{g}$ can be decomposed into $n+1$ subspaces
\begin{equation}
\mathfrak{g} = V_0 \oplus V_1 \oplus \ldots \oplus V_n 
\end{equation}
and satisfies the Weimar-Woods conditions \cite{WW1, WW2}
\begin{equation}
\left[V_p,V_q\right]\subset\bigoplus_{s\leq p+q}V_s,\quad p,q=0,1,\ldots,n .
\end{equation}

We will now discuss how to properly choose the subset partition of $S^{(\infty)}$ and we apply our method of infinite $S$-expansion with ideal subtraction in order to show that the generalized In\"{o}n\"{u}-Wigner contraction fits our scheme. We proceed methodically with the following steps:
\begin{enumerate}
\item We perform an infinite $S$-expansion with the infinite abelian semigroup $S^{(\infty)}$, endowed with the multiplication rule (\ref{prodsemigrinfinito}), on the original algebra $\mathfrak{g}$ satisfying the Weimar-Woods conditions.
\item We properly define subsets $S_p$ of $S^{(\infty)}$ such that they satisfy the resonant condition (\ref{groupdecomposition}) in the case in which the starting algebra $\mathfrak{g}$ satisfies the Weimar-Woods conditions, in order to be able to extract a resonant subalgebra from $\mathfrak{g}^\infty_S$ (by using Theorem \ref{Tres}). The resonant condition (\ref{groupdecomposition}) reads as follows when the Weimar-Woods conditions hold:
\begin{equation}\label{rescond}
S_p \cdot S_q \subset \bigcap_{r\leq p+q} S_r.
\end{equation}
\item We define a $S_p$ partition such that the conditions (\ref{intzero}) and (\ref{ressemigroup}) are fulfilled, in order to be able to extract a reduced algebra from the resonant subalgebra of the infinitely $S$-expanded one. This will be done through ideal subtraction. The reduced algebra will be our target, reproducing, in this way, a generalized In\"{o}n\"{u}-Wigner contraction, while the ideal subalgebra will be taken apart. 
\end{enumerate}
In this way, the whole procedure consisting in $S$-expanding with an infinite semigroup and subtracting the ideal will be able to reproduce a generalized In\"{o}n\"{u}-Wigner contraction.

We first of all perform the infinite $S$-expansion on $\mathfrak{g}$, obtaining the infinitely $S$-expanded algebra
\begin{equation}\label{sumanew}
\mathfrak{g}^\infty_S = S^{(\infty)} \times \mathfrak{g} =\left( \lbrace \lambda_\alpha \rbrace _0 ^\infty \times V_0 \right) \oplus \left( \lbrace \lambda_\alpha \rbrace_0 ^\infty \times V_1 \right) \oplus \ldots \oplus \left( \lbrace \lambda_\alpha \rbrace_0^\infty \times V_n \right).
\end{equation}

In order to proceed with the extraction of the infinite resonant subalgebra, we must split the semigroup $S^{(\infty)}$ into $n+1$ \textit{infinite} subsets $S_p$ such that, when $\mathfrak{g}$ satisfies the Weimar-Woods conditions (that is our case), the condition (\ref{rescond}) is fulfilled, where $S_p \cdot S_q$ denotes the set of all products of all elements of $S_p$ and all elements of $S_q$. Thus, we must define such a decomposition for the infinite semigroup $S^{(\infty)}$. 

Now, let
\begin{equation}\label{subsetdecinf}
S^{(\infty)} = \bigcup_{p=0}^n S_p
\end{equation}
be a subset decomposition of $S^{(\infty)}$, where the subsets $S_p \subset S^{(\infty)}$ are defined by
\begin{equation}
S_p = \lbrace \lambda_{\alpha_p}, \; \alpha_p = p, \ldots, \infty  \rbrace , \quad p = 0, \ldots n.
\end{equation}
This can also be visualized, for making it clearer, as
\begin{align}
S^{(\infty)} & = S_0 \cup S_1 \cup S_2 \cup \ldots \cup S_n = \nonumber \\
& = \lbrace \lambda_0 , \lambda_1 , \ldots , \lambda_\infty \rbrace \cup \lbrace \lambda_1 , \lambda_2, \ldots , \lambda_\infty \rbrace \cup \lbrace \lambda_2, \lambda_3 , \ldots, \lambda_ \infty \rbrace \cup \ldots \cup \lbrace \lambda_n , \lambda_{n+1}, \ldots , \lambda_\infty \rbrace.
\end{align}
The subset decomposition (\ref{subsetdecinf}) is a \textit{resonant} one under the semigroup product (\ref{prod}), since it satisfies (\ref{rescond}).
Thus, according to Theorem \ref{Tres}, the direct sum
\begin{equation}
\mathfrak{g}^\infty_R = \bigoplus_{p=0}^n W_p ,
\end{equation}
with 
\begin{equation}
W_p = S_p \times V_p ,
\end{equation}
is a \textit{resonant subalgebra} of $\mathfrak{g}^\infty_S$. 

Let us now consider the resonant subalgebra $\mathfrak{g}^\infty_R = \bigoplus_{p=0}^n W_p = \bigoplus_{p=0}^n S_p \times V_p$ and write the following $S_p$ partition: $S_p = \hat{S}_p \cup \check{S}_p$, where
\begin{equation}\label{partinf1}
\check{S}_p = \lbrace \lambda_{\alpha_p} , \; \alpha_p = p \rbrace  \equiv \lbrace \lambda_p \rbrace ,
\end{equation}
\begin{equation}\label{partinf2}
\hat{S}_p = \lbrace \lambda_{\alpha_p} , \; \alpha_p= p+1,\ldots , \infty  \rbrace .
\end{equation}
This $S_p$ partition satisfies
\begin{equation}\label{firstcond}
\hat{S}_p \cap \check{S}_p = \emptyset ,
\end{equation}
which is exactly the condition (\ref{intzero}).
The second condition which must be fulfilled in order to be able to extract a reduced algebra from the resonant subalgebra $\mathfrak{g}^\infty_S$ when the original algebra $\mathfrak{g}$ satisfies the Weimar-Woods conditions reads
\begin{equation}\label{secondcond}
\check{S}_p \cdot \hat{S}_q \subset \bigcap _{r \leq p+q} \hat{S}_r .
\end{equation} 
In the present case, $\check{S}_p$ and $\hat{S}_q$ are respectively given by
\begin{equation}
\check{S}_p = \lbrace \lambda_{p} \rbrace , 
\end{equation}
\begin{equation}
\hat{S}_q = \lbrace \lambda_{\alpha_q}, \; \alpha_q = q+1, \ldots , \infty \rbrace .
\end{equation}
Thus, the condition (\ref{secondcond}) is fulfilled, since, in this case,
\begin{equation}
\bigcap _{r=0}^{p+q} \hat{S}_r= \hat{S}_{p+q},
\end{equation}
where $\hat{S}_{p+q} = \lbrace \lambda_{p+q+m}, m =1, \ldots, \infty \rbrace$, and
\begin{equation}
\check{S}_p \cdot \hat{S}_q = \hat{S}_{p+q} ,
\end{equation}
where we have taken into account the product (\ref{prodsemigrinfinito}). We can thus conclude that the $S_p$ partition we have chosen satisfies the reduction condition, and we can now extract a reduced algebra from the resonant subalgebra $\mathfrak{g}^\infty_S$.

Indeed, what we have done by considering this particular $S_p$ partition induce, according with Theorem \ref{teoiza}, the following decomposition on the resonant subalgebra:
\begin{equation}
\mathfrak{g}^\infty_R = \check{\mathfrak{g}}_R \oplus \hat{\mathfrak{g}}^\infty_R ,
\end{equation}
where
\begin{equation}
\check{\mathfrak{g}}_R = \bigoplus_{p=0}^n \check{S}_p \times V_p , 
\end{equation}
\begin{equation}
\hat{\mathfrak{g}}^\infty_R = \bigoplus_{p=0}^n \hat{S}_p \times V_p.
\end{equation}
Let us observe that, according with what we have previously observed when describing the method in general, $\check{\mathfrak{g}}_R$ is \textit{finite}, since it is the direct sum of products between finite subsets and finite subspaces, while $\hat{\mathfrak{g}}^\infty_R$ is \textit{infinite}, due to the fact that the $\hat{S}_p$'s are infinite (sub)sets.

We can now write
\begin{equation}
\hat{\mathfrak{g}}^\infty_R = \bigoplus_{p=0}^n \hat{W}_p ,
\end{equation}
with
\begin{equation}
\hat{W}_p = \hat{S}_p \times V_p = S_{p+1} \times V_p ,
\end{equation}
where $S_{p+1}= \lbrace \lambda_{p+m}, \; m=1, \ldots , \infty \rbrace$. One can now easily prove that, by construction, we have
\begin{equation}
\left[\hat{\mathfrak{g}}^\infty_R , \hat{\mathfrak{g}}^\infty_R \right] \subset \hat{\mathfrak{g}}^\infty_R ,
\end{equation}
that is to say, $\hat{\mathfrak{g}}^\infty_R$ is an infinite subalgebra of $\mathfrak{g}^\infty_R$ (and, consequently, an infinite subalgebra of $\mathfrak{g}^\infty_S$). In particular, it is an \textit{ideal subalgebra}, since it also satisfies 
\begin{equation}
\left[\check{\mathfrak{g}}_R , \hat{\mathfrak{g}}^\infty_R \right] \subset \hat{\mathfrak{g}}^\infty_R .
\end{equation}
This allows us to write
\begin{equation}
\check{\mathfrak{g}}_R = \mathfrak{g}^\infty_R \ominus \mathcal{I},
\end{equation}
where we have denoted by $\mathcal{I}$ the infinite ideal subalgebra, $\mathcal{I}\equiv \hat{\mathfrak{g}}^\infty_R$, and where, applying Theorem \ref{TeoRedIdeal}, the algebra $\check{\mathfrak{g}}_R$ corresponds to a reduced algebra.


Since, as stated in \cite{Izaurieta} and previously reviewed in the present paper, the generalized In\"{o}n\"{u}-Wigner contraction corresponds to the reduction of a resonant subalgebra of the $S$-expanded one, we have thus reproduced a generalized In\"{o}n\"{u}-Wigner contraction by performing our method of infinite $S$-expansion with subsequent ideal subtraction.

We can summarize our approach saying that, after having performed an infinite $S$-expansion with the semigroup $S^{(\infty)}$ on an original Lie algebra $\mathfrak{g}$, we can reproduced the generalized 
In\"{o}n\"{u}-Wigner contraction by first of all extracting a resonant subalgebra $\mathfrak{g}^\infty_R$ from the infinitely $S$-expanded algebra $\mathfrak{g}^\infty_S = S^{(\infty)} \times \mathfrak{g}$, and by then subtracting an infinite ideal subalgebra 
from the resonant one. In this way, the procedure reproduce the result of a reduction of a resonant subalgebra of an (infinitely) $S$-expanded algebra, which, according with \cite{Izaurieta}, lets the generalized In\"{o}n\"{u}-Wigner contraction fit the $S$-expansion scheme. In this way, we have obtained an alternative view of the generalized In\"{o}n\"{u}-Wigner contraction process.

We will discuss some examples of application of our method in Section \ref{Examples}. We now move to the description of the way for finding the invariant tensors of the (super)algebras obtained with our method in terms of those of the original (super)algebras.

\subsection{Invariant tensors of (super)algebras obtained through infinite $S$-expansion with ideal subtraction}  
          
In Ref. \cite{Izaurieta}, the authors developed a theorem which describes how to write the components of the invariant tensor of a target algebra obtained through a (finite) $S$-expansion in terms of those of the initial algebra (precisely, we are referring to Theorem VII.1 of Ref. \cite{Izaurieta}). Then, in Theorem VII.2 of the same paper, they have given an expression for the invariant tensor for a $0_S$-reduced algebra.

Taking into account this result, we can now show how to write the invariant tensors for the (super)algebras which can be obtained by applying our method of infinite $S$-expansion with ideal subtraction.

\begin{theorem}\label{teo}
Let $\mathfrak{g}$ be a Lie (super)algebra of basis $\lbrace T_A \rbrace$ and let $\langle T_{A_0}\ldots T_{A_N}\rangle$ be an invariant tensor for $\mathfrak{g}$. 
Let $\mathfrak{g}^\infty_S = S^{(\infty)}\times \mathfrak{g}$ be the infinite $S$-expanded (super)algebra obtained using the infinite abelian semigroup $S^{(\infty)} = \lbrace \lambda_\alpha , \; \alpha = 0 , \ldots , \infty \rbrace$.
Let $\mathfrak{g}^\infty_R$ be an infinite resonant subalgebra of $\mathfrak{g}^\infty_S$ and let $\mathcal{I}$ be an infinite ideal subalgebra of $\mathfrak{g}^\infty_R$.
Then, 
\begin{equation}\label{topologicalInvariant}
\langle T^{\alpha_{p_0}}_{A_{p_0}}\ldots T^{\alpha_{p_N}}_{A{p_N}}\rangle =  \alpha^{m}\delta^{\alpha_{p_0} + \alpha_{p_1}+\alpha_{p_2}+\ldots + \alpha_{p_N}}_m \langle T_{A_0}\ldots T_{A_N}\rangle ,
\end{equation}
where $\alpha^{m}$ are arbitrary constants, 
corresponds to an invariant tensor for the finite (super)algebra
\begin{equation}
\check{\mathfrak{g}}_R = \mathfrak{g}^\infty_R \ominus \mathcal{I},
\end{equation}
having denoted the generators of $\check{\mathfrak{g}}_R$ by $\lambda_{\alpha_{p_i}} T_{A_{p_i}}\equiv T^{\alpha_{p_i}}_{A_{p_i}}$, with $i=0, \ldots, N $, and where the set $\lbrace \lambda_{\alpha_p} \rbrace$ is finite. 
\end{theorem}

\begin{proof}
The proof of Theorem \ref{teo} can be developed by applying Theorem \ref{TeoRedIdeal} of the present paper. Indeed, as stated in Theorem \ref{TeoRedIdeal}, we have that the subtraction of an infinite ideal subalgebra from an infinite resonant subalgebra of an infinitely $S$-expanded (super)algebra (using the semigroup $S^{(\infty)}$ on the original (super)algebra) corresponds to a reduction. In particular, we have seen that it reproduces the same result of a $0_S$-reduction. 

In this way, one can write the invariant tensor of the (super)algebra obtained with our method of infinite $S$-expansion with ideal subtraction by applying Theorem VII.2 of Ref. \cite{Izaurieta}, which, indeed, gives an expression for the invariant tensor for a $0_S$-reduced algebra.

Thus, it is straightforward to show that the invariant tensor for the (super)algebra $\check{\mathfrak{g}}_R=\mathfrak{g}_R^\infty\ominus\mathcal{I}$ can be written in the form
\begin{equation}
\langle T^{\alpha_{p_0}}_{A_{p_0}}\ldots T^{\alpha_{p_N}}_{A{p_N}}\rangle =  \alpha^{m}\delta^{\alpha_{p_0} + \alpha_{p_1}+\alpha_{p_2}+\ldots + \alpha_{p_N}}_m \langle T_{A_0}\ldots T_{A_N}\rangle ,
\end{equation}
being $\alpha^{m}$ arbitrary constants, where we have denoted the generators of $\check{\mathfrak{g}}_R$ by $\lambda_{\alpha_{p_i}} T_{A_{p_i}}\equiv T^{\alpha_{p_i}}_{A_{p_i}}$, with $i=0, \ldots, N$, and where the set $\lbrace \lambda_{\alpha_p} \rbrace$ is finite. 
\end{proof}

\section{Examples of application} \label{Examples}

In the following, we develop some examples of application, in which we replicate the effects of both standard and generalized In\"on\"u-Wigner contractions of some (super)algebras presented in the literature through our infinite $S$-expansion approach with ideal subtraction. 

We also find the invariant tensors of some of the mentioned (super)algebras by applying Theorem \ref{teo}.

\subsection{Symmetric cosets}

As shown in Section \ref{Review}, the standard In\"{o}n\"{u}-Wigner contraction can be applied to symmetric cosets of simple Lie algebras, \textit{i.e.} to cosets of Lie algebras which can be written as $\mathfrak{g}=\mathfrak{h}\oplus\mathfrak{p}$, where the following commutation relations hold:
\begin{align}
 \left[\mathfrak{h},\mathfrak{h}\right]\subset&\mathfrak{h},\nonumber\\
 \left[\mathfrak{h},\mathfrak{p}\right]\subset&\mathfrak{p},\\
  \left[\mathfrak{p},\mathfrak{p}\right]\subset&\mathfrak{h}.\nonumber
\end{align}
After having performed a standard In\"{o}n\"{u}-Wigner contraction ($\mathfrak{h}'= \mathfrak{h}'$, $\mathfrak{p}'= \varepsilon \mathfrak{p}$, and $\varepsilon \longrightarrow 0$) on $\mathfrak{g}$, we get the contracted commutation relations
\begin{align}
\left[\mathfrak{h}',\mathfrak{h}'\right]\subset&\mathfrak{h}',\nonumber\\
\left[\mathfrak{h}',\mathfrak{p}'\right]\subset&\mathfrak{p}',\label{commrelend}\\
\left[\mathfrak{p}',\mathfrak{p}'\right]=&0.\nonumber
\end{align}
The standard In\"{o}n\"{u}-Wigner contraction abelianizes the last commutator ($\left[\mathfrak{p}',\mathfrak{p}'\right]=0$), and thus produces an algebra that is non-isomorphic to the starting one. \footnote{Let us observe that in the case in which $\left[\mathfrak{p},\mathfrak{p}\right]=\mathfrak{h}\oplus\mathfrak{p}$,
that is to say, when we are taking into account a non-symmetric coset, the standard In\"{o}n\"{u}-Wigner contraction still abelianizes this commutator.}

With the method presented in this paper, we can reproduce the same result.
Indeed, the case of the standard In\"{o}n\"{u}-Wigner contraction from an algebra $\mathfrak{g}=\mathfrak{h}\oplus \mathfrak{p}$, where $\mathfrak{h}$ is a subalgebra and $\mathfrak{p}$ a symmetric coset, to an algebra $\check{\mathfrak{g}}_R= \mathfrak{h}'\ltimes \mathfrak{p}'=\mathfrak{h}\ltimes \mathfrak{p}'$ satisfying (\ref{commrelend}) can be reproduced by performing our prescription of infinite $S$-expansion with ideal subtraction. Indeed, after having infinitely $S$-expanded the original algebra using the semigroup $S^{(\infty)}$, we can write
\begin{equation}
\check{\mathfrak{g}}_R = \mathfrak{g}^\infty_R \ominus \mathcal{I},
\end{equation}
where
\begin{align}
\mathfrak{g}^\infty_R &= \lbrace \lambda_{2l} , \; l = 0, \ldots, \infty \rbrace \times \mathfrak{h} \oplus \lbrace \lambda_{2l+1} , \; l = 0, \ldots, \infty \rbrace \times \mathfrak{p} , \nonumber \\
 \mathcal{I}&= \lbrace \lambda_{2m} , \; m = 1, \ldots, \infty\rbrace \times \mathfrak{h} \oplus \lbrace \lambda_{2m+1} , \; m = 1, \ldots, \infty \rbrace \times \mathfrak{p} , \nonumber \\
 \check{\mathfrak{g}}_R&=\lbrace \lambda_0 \rbrace \times \mathfrak{h}\oplus \lbrace \lambda_1 \rbrace \times \mathfrak{p} = \mathfrak{h}' \oplus \mathfrak{p} ' .
\end{align}

\subsubsection{From the $AdS$ to the Poincar\'{e} algebra in $D=3$}

Let us see the particular feature described above through an explicit, physical example.

The Poincar\'e algebra $\mathfrak{iso}(D-1,1)$ can be obtained as a standard In\"{o}n\"{u}-Wigner contraction of the anti-de Sitter ($AdS$) algebra $ \mathfrak{so}(D-1,2)$.

We now consider the $AdS$ algebra in three dimensions, whose commutation relations read:
\begin{align}
\left[\tilde{J}_{a},\tilde{J}_{b}\right]=&\epsilon_{abc}\tilde{J}_{c},\nonumber\\
\left[\tilde{J}_{a},\tilde{P}_{b}\right]=& \epsilon_{abc}\tilde{P}_{c},\label{x1}\\
\left[\tilde{P}_{a},\tilde{P}_{b}\right]=&\epsilon_{abc}\tilde{J}_{c}.\nonumber
\end{align}
Then, we apply the infinite $S$-expansion procedure with ideal subtraction described in the present work, considering $V_0 = \lbrace \tilde{J}_a \rbrace$ and $V_1 = \lbrace \tilde{P}_a \rbrace$. After having infinitely $S$-expanded the original algebra with the semigroup $S^{(\infty)}$, we write 
\begin{align}
\mathfrak{g}^\infty_R &= \lbrace \lambda_{2l}, \; l =0, \ldots, \infty \rbrace \times V_0 \oplus \lbrace  \lambda_{2l+1}, \; l =0, \ldots, \infty  \rbrace \times V_1 , \nonumber \\
 \mathcal{I}&= \lbrace  \lambda_{2m}, \; m =1, \ldots, \infty  \rbrace \times V_0 \oplus \lbrace  \lambda_{2m+1}, \; m =1, \ldots, \infty  \rbrace \times V_1 , \nonumber \\
 \check{\mathfrak{g}}_R&=\lbrace \lambda_0 \rbrace \times V_0 \oplus \lbrace \lambda_1 \rbrace \times V_1 .
\end{align}
Subsequently, we rename the generators as follows:
\begin{align}
 \lambda_0 \tilde{J}_{a}=&J_{a},\\
 \lambda_1 \tilde{P}_{a}=&P_a. 
\end{align}
In this way, we obtain the Poincar\'e algebra $\mathfrak{iso}(2,1)$ in three dimensions, which can be written in terms of the following commutation relations:
\begin{align}
    \left[J_{ab},J_{cd}\right]=&\eta_{bc}J_{ad}-\eta_{ac}J_{bd}-\eta_{bd}J_{ac}+\eta_{ad}J_{bc},\nonumber\\
  \left[J_{ab},P_{c}\right]=& \eta_{bc}P_{a}-\eta_{ac}P_{b},\label{poincarealgebrak=1}\\
 \left[P_{a},P_{b}\right]=& 0.\nonumber
\end{align}
We have thus reached a Poincar\'e algebra from an $AdS$ algebra by performing an infinite $S$-expansion with ideal subtraction, reproducing, in this way, the effects of a standard In\"on\"u-Wigner contraction.

\subsection{From the $AdS$ to the Maxwell-like algebra $\mathcal{M}_5$}

The $AdS$ algebra $ \mathfrak{so}(D-1,2)$ has the set of generators $\left\{\tilde{J}_{ab},\tilde{P}_{a}\right\}$, and it is endowed with the following commutation relations:
\begin{align}
    \left[\tilde{J}_{ab},\tilde{J}_{cd}\right]=&\eta_{bc}\tilde{J}_{ad}-\eta_{ac}\tilde{J}_{bd}-\eta_{bd}\tilde{J}_{ac}+\eta_{ad}\tilde{J}_{bc},\nonumber \\
  \left[\tilde{J}_{ab},\tilde{P}_{c}\right]=& \eta_{bc}\tilde{P}_{a}-\eta_{ac}\tilde{P}_{b}, \nonumber \\
 \left[\tilde{P}_{a},\tilde{P}_{b}\right]=& \tilde{J}_{ab}. \label{adscommrel}
\end{align}

The Maxwell algebras type $\mathcal{M}_m$ (alternatively known as generalized Poincar\'{e} algebras, see Ref. \cite{CR2}) can be obtained as a finite $S$-expansion (with resonance and reduction) from the anti-de Sitter algebra $\mathfrak{so}(D-1,2)$, using semigroups of the type $S^{(N)}_E=\left\{\lambda_\alpha\right\}^{N+1}_{\alpha=0}$, which are endowed with the multiplication rules $\lambda_\alpha\lambda_\beta=\lambda_{\alpha+\beta}$ if $\alpha+\beta\leq N + 1$, and $\lambda_\alpha\lambda_\beta=\lambda_{N+1}$ if $\alpha+\beta > N + 1$.

The reduction involved in the procedure takes into account the presence of a zero element in the semigroup. This zero element is defined as $\lambda_{0_S}=\lambda_{N+1}$, and depends of the number $N$ in the semigroup.

We can now reach the Maxwell-like algebra $\mathcal{M}_5$ by performing our method of infinite $S$-expansion with ideal subtraction, starting from the $AdS$ algebra.

We follow the notation and the subspaces partition adopted in Ref. \cite{Concha1} and perform an infinite $S$-expansion with the infinite semigroup $S^{(\infty)}= \left\{\lambda_\alpha\right\}^{\infty}_{\alpha=0}$, obtaining
\begin{equation}
\mathfrak{g}_S^\infty=\left\{\lambda_\alpha\right\}^{\infty}_{\alpha=0}\times \left\{\tilde{J}_{ab},\tilde{P}_{a}\right\} = \left(\left\{\lambda_\alpha\right\}^{\infty}_{\alpha=0}\times V_0\right) \oplus \left( \left\{\lambda_\alpha\right\}^{\infty}_{\alpha=0}\times V_1 \right),
\end{equation}
where $V_0 = \lbrace \tilde{J}_{ab} \rbrace$ and $V_1 = \lbrace \tilde{P}_a \rbrace$.

Now, following our method, we can write (adopting the subspace partition of Ref. \cite{CR2}):
\begin{align}
\mathfrak{g}^\infty_R &= \lbrace \lambda_{2l}, \; l =0, \ldots, \infty \rbrace \times V_0 \oplus \lbrace  \lambda_{2l+1}, \; l =0, \ldots, \infty  \rbrace \times V_1 , \nonumber \\
 \mathcal{I}&= \lbrace  \lambda_{2m}, \; m =2, \ldots, \infty  \rbrace \times V_0 \oplus \lbrace  \lambda_{2m+1}, \; m =2, \ldots, \infty  \rbrace \times V_1 , \nonumber \\
 \check{\mathfrak{g}}_R&=\lbrace \lambda_0 , \lambda_2 \rbrace \times V_0 \oplus \lbrace \lambda_1 , \lambda_3 \rbrace \times V_1 .
\end{align}

Thus, after having renamed the generators in the following way: $\lambda_ 0 \tilde{J}_{ab}=J_{ab}$, $\lambda_1 \tilde{P}_{a}=P_{a}$, $\lambda_2 \tilde{J}_{ab}=Z_{ab}$, and $\lambda_3 \tilde{P}_{a}=Z_{a}$, we finally end up with the Maxwell-like algebra $\mathcal{M}_5$, which is indeed endowed with the commutation relations
\begin{align}
\left[J_{ab},J_{cd}\right]=&\eta_{bc}J_{ab}-\eta_{ac}J_{bd}-\eta_{bd}J_{ac}+\eta_{ad}J_{bc},\\
\left[Z_{ab},J_{cd}\right]=&\eta_{bc}Z_{ad}-\eta_{ac}Z_{bd}-\eta_{bd}Z_{ac}+\eta_{ad}Z_{bc},\\
  \left[Z_{ab},P_{c}\right]=& \eta_{bc}Z_{a}-\eta_{ac}Z_{b},\\
 \left[J_{ab},Z_{b}\right]=&\eta_{bc}Z_{a}-\eta_{ac}Z_{b},\\
\left[J_{ab},P_{c}\right]=& \eta_{bc}P_{a}-\eta_{ac}P_{b},\\
 \left[P_{a},P_{b}\right]=& Z_{ab},\\
 \left[Z_{ab},Z_{cd}\right]=&\left[Z_{a},Z_{b}\right]=0.
\end{align}
We have thus reached the Maxwell-like algebra $\mathcal{M}_5$ by performing an infinite $S$-expansion with ideal subtraction on an $AdS$ algebra.

\paragraph{Invariant tensor of the Maxwell-like algebra in $D=3$}\

In the three-dimensional case, considering equation (\ref{topologicalInvariant}), the components of the invariant tensor of the Maxwell-like algebra different from zero are the following ones:
\begin{align}
&\langle J_{ab} J_{cd}\rangle=\alpha^{0} \left(\eta_{ad}\eta_{bc}-\eta_{ac}\eta_{bd}\right),&
\langle J_{ab}  P_{c}\rangle=\alpha^{1} \epsilon_{abc},\\
&\langle Z_{ab} J_{cd}\rangle=\alpha^{2} \left(\eta_{ad}\eta_{bc}-\eta_{ac}\eta_{bd}\right),&
\langle J_{ab}  Z_{c}\rangle=\alpha^{3} \epsilon_{abc},\\
&\langle Z_{ab}  P_{c}\rangle=\alpha^{3} \epsilon_{abc},&
\langle P_{a}  P_{b}\rangle \ =\alpha^{2}\eta_{ab},
\end{align}
where $\alpha^m$, $m=0,1,2,3$, are arbitrary constants.
One could now construct a Lagrangian for a gravitational theory, as it was done in the literature (see Refs. \cite{Diego, Seba}). 

\subsection{Bargmann algebra and Newton-Hooke algebra}

In the following, we apply our method involving an infinite $S$-expansion with subsequent ideal subtraction in order to relate different algebras which have been objects of great interest in the literature (see Refs. \cite{Rosseel,galalg,Andriga,Andriga2,Gibbons,frodo,frodo2, Yves,Papageorgiou} for further details), namely the Poincar\'e and Galilean algebras, the Bargmann algebra, and Newton-Hooke algebra.

The Bargmann algebra $\mathfrak{b}(D-1,1)$ is the Galilean algebra \cite{galalg, Yves} augmented with a central generator $M$, \footnote{Namely, a generator that commutes with all other generators of the algebra. In $D = 3$, three such central generators can be introduced.} and can be obtained by performing a contraction on $\mathfrak{iso}(D-1,1)\oplus\mathfrak{g}_M$, where $\mathfrak{iso}(D-1,1)$ is the Poincar\'e algebra in $D$ dimensions, and where $\mathfrak{g}_M$ is a commutative subalgebra \cite{Andriga} spanned by a central generator $M$.

Some algebras involving the presence of the cosmological constant were deeply studied in the past years. 
The most relevant symmetry groups with the presence of a cosmological constant $\Lambda$ are the de Sitter ($dS$) and the anti-de Sitter ($AdS$) groups, that are related to relativistic symmetries which also involved the velocity of the light in the vacuum, $c$. In the non-relativistic limit, that is to say, when $c\longrightarrow\infty$ and $\Lambda\longrightarrow0$, with $-c^2\Lambda$ finite, the (centrally extended) $dS$ and the $AdS$ Lie algebras become centrally extended version of the so-called Newton-Hooke algebra, which describes Galilean physics with a cosmological constant $\lambda=-c^2\Lambda$ \cite{Papageorgiou}.
The central extensions of the Newton-Hooke Lie algebras can be obtained as contractions of trivial central extensions of the $dS$ and $AdS$ Lie algebras (see Ref. \cite{Papageorgiou}).

\subsubsection{The Bargmann algebra}

Let us now consider a central extension (with central generator $M$) of the Poincar\'e algebra in $\mathfrak{iso}(D-1,1)$ (this extension of the Poincar\'e algebra is a commutative algebra $\mathfrak{g}_M$ spanned by the central generator $M$), namely $\mathfrak{iso}(D-1,1) \oplus \mathfrak{g}_M$, in order to apply our procedure of infinite $S$-expansion with ideal subtraction and obtain the Bargmann algebra in three dimensions, $\mathfrak{b}(D-1,1)$. In this way, we will reproduce the result of a contraction with our method.

We thus start by redefining (according with the notation of \cite{Andriga}) the generators of the starting algebra $\mathfrak{iso}(D-1,1) \oplus \mathfrak{g}_M$ as follows:
\begin{equation}
J_{ij} = J_{ij}, \quad P_i = P_i, \quad J_{i0} = G_i , \quad P_0 = H \rightarrow H+M.
\end{equation}
We then write the following partition over subspaces:
\begin{align}
V_0 = &\lbrace{J_{ij} , H\rbrace}, \\
V_1 = &\lbrace{P_i,G_i \rbrace}, \\
V_2 = &\lbrace{M \rbrace}.
\end{align}
After that, we perform an infinite $S$-expansion with $S^{(\infty)}$, and, following the procedure described in Subsection \ref{infexp}, we consider the following resonant subset decomposition of the semigroup $S^{(\infty)}$:
\begin{align}
 S_0=&\left\{\lambda_0\right\}\cup\left\{\lambda_{1}, \ldots , \lambda_\infty \right\},\nonumber\\
S_1=&\left\{\lambda_1\right\}\cup\left\{\lambda_{2}, \ldots , \lambda_\infty \right\} , \\
S_2=&\left\{\lambda_2\right\}\cup\left\{\lambda_{3}, \ldots , \lambda_\infty \right\}.\nonumber
\end{align}
The resonant subalgebra $\mathfrak{g}^\infty_R=\check{\mathfrak{g}}_R \oplus \hat{\mathfrak{g}}^\infty_R$ results to be given by the direct sum of the following terms:
\begin{align}
 S_0\times V_0&=\lbrace\lambda_0 \rbrace \times V_0 \oplus \left\{\lambda_{1} , \ldots , \lambda_\infty\right\} \times V_0,\\
 S_1\times V_1&= \lbrace\lambda_1 \rbrace \times  V_1 \oplus \left\{\lambda_{2}, \ldots , \lambda_\infty \right\} \times V_1,\\
 S_2\times V_2&= \lbrace \lambda_2 \rbrace \times V_2 \oplus \left\{\lambda_3 , \ldots , \lambda_\infty \right\} \times V_2.
\end{align}

Then, the ideal reads
\begin{align}
 \mathcal{I}=&\left(\left\{\lambda_{1} , \ldots , \lambda_\infty\right\}\times V_0\right)\oplus\left(\left\{\lambda_{2} , \ldots , \lambda_\infty\right\}\times V_1\right)\oplus\left(\left\{\lambda_{3} , \ldots , \lambda_\infty\right\}\times V_2\right),
\end{align}
and we can perform the ideal subtraction on the resonant subalgebra $\mathfrak{g}^\infty_R$.

Thus, the target algebra $\check{\mathfrak{g}}_R$, which is, in this case, the Bargmann algebra $\mathfrak{b}(D-1,1)$, is given by $\check{\mathfrak{g}}_R=\mathfrak{g}_R^\infty\ominus\mathcal{I}$.

Indeed, after having properly renamed the generator as follows: $\lambda_0 J_{ij} = \tilde{J}_{ij}$, $\lambda_0 H = \tilde{H}$, $\lambda_1 P_i = \tilde{P}_i$, $\lambda_1 G_i = \tilde{G}_i$, and $\lambda_2 M = \tilde{M}$, we finally get
\begin{align}
 \left[\tilde{J}_{ij},\tilde{J}_{kl}\right]=& 4 \delta_{[i[k}\tilde{J}_{l]j]}, \nonumber \\
 \left[\tilde{J}_{ij},\tilde{G}_k\right]=&-2 \delta_{k[i}\tilde{G}_{j]} , \nonumber \\
  \left[\tilde{J}_{ij},\tilde{P}_k\right]=&-2 \delta_{k[i} \tilde{P}_{j]}, \label{Bargmannalgebra}\\
  \left[\tilde{G}_{i},\tilde{P}_{j}\right]=&-\delta_{ij}\tilde{M}, \nonumber \\
  \left[\tilde{G}_{i},\tilde{H}\right]=&-\tilde{P}_i , \nonumber  
\end{align}
that are precisely the commutation relations of the Bargmann algebra $\mathfrak{b}(D-1,1)$.
Let us also observe that for $M=0$ this is the Galilean algebra.

\subsubsection{Centrally extended Newton-Hooke algebra in $D=3$}

We now perform an infinite $S$-expansion with ideal subtraction on a central extension (with two central charges) of the $AdS$ algebra in $D=3$, in order to reach a centrally extended Newton-Hooke algebra in three dimensions (see Ref. \cite{Papageorgiou} for further details on the Newton-Hooke algebras). In this way, we will reproduce the result of a contraction by following our method.

In $D=3$, we can write the three-dimensional $AdS$ algebra as
\begin{align} 
   \left[J_{a},J_{b}\right]=&\epsilon_{abc}J^{c},\nonumber\\
  \left[J_{a},P_{b}\right]=&\epsilon_{abc}P^{c},\label{ads}\\
  \left[P_{a},P_{b}\right]=& \epsilon_{abc}J^{c}.\nonumber
\end{align}
We now rename the generators, using the two-dimensional epsilon symbol with non-zero entries $\epsilon_{12}=-\epsilon_{21}=1$, as follows:
\begin{align}\label{changeofbasis}
\bar{J}=&-J_0,\nonumber\\
K_i=&-\epsilon_{ij}J_j,\\
\bar{H} =&-P_0.\nonumber
\end{align}
The generator $K_i$, $i = 1, 2$, generate boosts in the $i$-th spatial direction.
Then, the algebra (\ref{ads}) reads
\begin{align}
  \left[K_{i},K_{j}\right]=&\epsilon_{ij}\bar{J},\nonumber\\
  \left[K_{i},P_{j}\right]=&\delta_{ij}\bar{H},\nonumber\\
  \left[P_{i},P_{j}\right]=&\epsilon_{ij}\bar{J},\nonumber\\
    \left[K_{i},\bar{J}\right]=&\epsilon_{ij}K_j,\label{adschangebasis}\\
  \left[K_{i},\bar{H}\right]=&P_i,\nonumber\nonumber\\
  \left[\bar{H},P_i\right]=& K_i,\nonumber\\
  \left[P_i,\bar{J}\right]=&\epsilon_{ij}P_j.\nonumber
\end{align}

We now consider two central extensions: $S$ and $M$. Then, according with \cite{Papageorgiou}, we define
\begin{equation}\label{basetildeNH1}
 H=\bar{H}-M, \quad J=\bar{J}-S,
\end{equation}
namely $\bar{H}=H+M$ and $\bar{J}=J+S$,
and we perform the following subspaces partition:
\begin{align}
V_0 =& \lbrace{J,H \rbrace}, \\
V_1 =& \lbrace{P_i,K_i \rbrace}, \\
V_2 =& \lbrace{M,S \rbrace} .
\end{align}

After that, applying the procedure developed in Subsection \ref{infexp}, we perform an infinite $S$-expansion with $S^{(\infty)}$ and we consider the following resonant subset decomposition of the semigroup $S^{(\infty)}$:
\begin{align}
 S_0=&\left\{\lambda_0\right\}\cup\left\{\lambda_{1}, \ldots , \lambda_\infty \right\},\nonumber\\
S_1=&\left\{\lambda_1\right\}\cup\left\{\lambda_{2}, \ldots , \lambda_\infty \right\} , \\
S_2=&\left\{\lambda_2\right\}\cup\left\{\lambda_{3}, \ldots , \lambda_\infty \right\}.\nonumber
\end{align}
The resonant subalgebra $\mathfrak{g}^\infty_R=\check{\mathfrak{g}}_R \oplus \hat{\mathfrak{g}}^\infty_R$ results to be given by the direct sum of the following terms:
\begin{align}
 S_0\times V_0&=\lbrace\lambda_0 \rbrace \times V_0 \oplus \left\{\lambda_{1} , \ldots , \lambda_\infty\right\} \times V_0,\\
 S_1\times V_1&= \lbrace\lambda_1 \rbrace \times  V_1 \oplus \left\{\lambda_{2}, \ldots , \lambda_\infty \right\} \times V_1,\\
 S_2\times V_2&= \lbrace \lambda_2 \rbrace \times V_2 \oplus \left\{\lambda_3 , \ldots , \lambda_\infty \right\} \times V_2.
\end{align}

Then, the ideal reads
\begin{align}
 \mathcal{I}=&\left(\left\{\lambda_{1} , \ldots , \lambda_\infty\right\}\times V_0\right)\oplus\left(\left\{\lambda_{2} , \ldots , \lambda_\infty\right\}\times V_1\right)\oplus\left(\left\{\lambda_{3} , \ldots , \lambda_\infty\right\}\times V_2\right),
\end{align}
and we can perform the ideal subtraction on the resonant subalgebra $\mathfrak{g}^\infty_R$.

The commutators of $\check{\mathfrak{g}}_R=\mathfrak{g}_R^\infty\ominus\mathcal{I}$ can be finally written as follows: 
\begin{align}
  \left[\tilde{K}_{i},\tilde{K}_{j}\right]=&\epsilon_{ij}\tilde{S},\nonumber\\
  \left[\tilde{K}_{i},\tilde{P}_{j}\right]=&\delta_{ij}\tilde{M},\nonumber\\
  \left[\tilde{P}_{i},\tilde{P}_{j}\right]=&\epsilon_{ij}\tilde{S},\nonumber\\
  \left[\tilde{K}_{i},\tilde{J}\right]=&\epsilon_{ij}\tilde{K}_j,\label{nh}\\
  \left[\tilde{K}_{i},\tilde{H}\right]=&\tilde{P}_i,\nonumber\\
  \left[\tilde{H},\tilde{P}_i\right]=&\tilde{K}_i,\nonumber\\
  \left[\tilde{P}_i,\tilde{J}\right]=&\epsilon_{ij}\tilde{P}_j,\nonumber
\end{align}
where we have defined $\lambda_0 J=\tilde{J}$, $\lambda_0 H = \tilde{H}$, $\lambda_1 P_i = \tilde{P}_i$, $\lambda_1 K_i = \tilde{K}_i$, $\lambda_2 M = \tilde{M}$, and $\lambda_2 S = \tilde{S}$.
These last commutation relations are those of a centrally extended version of the so-called Newton-Hooke algebra in three dimensions (see \cite{Papageorgiou} and references therein).

We have thus reached a centrally extended version of the Newton-Hooke algebra in $D=3$ by performing an infinite $S$-expansion with subsequent ideal subtraction starting from a central extension of the three-dimensional $AdS$ algebra.

\subsection{Non-standard Maxwell superalgebra in $D=3$ from the $AdS$ algebra $\mathfrak{osp}(2|1)\otimes \mathfrak{sp}(2)$}\label{nsmax}

The Maxwell superalgebra can be obtained as an In\"{o}n\"{u}-Wigner contraction of the $AdS$-Lorentz superalgebra, which is an $S$-expansion of the $AdS$ Lie algebra (see Ref. \cite{Seba} for further details). 

The non-standard Maxwell algebra \cite{Soroka2,Soroka3} can be recovered from the supersymmetric extension of the $AdS$-Lorentz algebra, by performing a suitable In\"{o}n\"{u}-Wigner contraction (see Ref. \cite{Lukierski}). 

In the following, we reproduce the non-standard Maxwell superalgebra in three dimensions, by performing an infinite $S$-expansion with $S^{(\infty)}$ on the $AdS$ superalgebra, and by subsequently removing an ideal. We also write the components of the invariant tensor of the target superalgebra in terms of those of the $AdS$ superalgebra. 

Let us thus consider the $AdS$ superalgebra in three dimensions, $\mathfrak{osp(2|1)}\otimes\mathfrak{sp(2)}$, generated by $\tilde{J}_{ab}$, $\tilde{P}_{a}$, and $\tilde{Q}_\alpha$, which satisfy the following commutation relations:
\begin{align}
 \left[\tilde{J}_{ab},\tilde{J}_{cd}\right]=&\eta_{bc}\tilde{J}_{ad}-\eta_{ac}\tilde{J}_{bd}-\eta_{bd}\tilde{J}_{ac}+\eta_{ad}\tilde{J}_{bc},\\
 \left[\tilde{J}_{ab}\tilde{P}_c\right]=&\eta_{bc}\tilde{P}_a-\eta_{ac}\tilde{P}_b,\\
 \left[\tilde{P}_a,\tilde{P}_b\right]=&\tilde{J}_{ab},\\
 \left[\tilde{P}_{a},\tilde{Q}_\alpha\right]=&\frac{1}{2}\left(\Gamma_a\tilde{Q}\right)_\alpha,\\
 \left[\tilde{J}_{ab},\tilde{Q}_\alpha\right]=&\frac{1}{2}\left(\Gamma_a\tilde{Q}\right)_\alpha,\\
 \left\{\tilde{Q}_\alpha,\tilde{Q}_\beta\right\}=&-\frac{1}{2}\left[\left(\Gamma^{ab}C\right)_{\alpha\beta}\tilde{J}_{ab}-2\left(\Gamma^aC\right)_{\alpha\beta}\tilde{P}_a\right].
\end{align}

We perform the following splitting of the $AdS$ superalgebra into two subspaces: $V_0 =\lbrace \tilde{J}_{ab}\rbrace $ and $V_1= \lbrace \tilde{P}_a ,\tilde{Q}_\alpha \rbrace$.
This subspace structure satisfies
\begin{align}
& [V_0, V_0] \subset V_0 , \nonumber \\ 
& [V_0, V_1] \subset V_1 , \nonumber \\ 
& [V_1, V_1] \subset V_0 \oplus V_1 .
\end{align}
We then consider the abelian semigroup $S^{(\infty)}$ endowed with the product (\ref{prodsemigrinfinito}) and we obtain an infinite-dimensional superalgebra as an infinite $S$-expansion of $\mathfrak{osp(2|1)}\otimes\mathfrak{sp(2)}$, using $S^{(\infty)}$.
Thus, a resonant partition of the semigroup $S^{(\infty)} = S_0 \cup S_1$ with respect to the product (\ref{prodsemigrinfinito}) is given by
\begin{equation}
S_p = \lbrace \lambda_{2m+p}, \; m = 0, \ldots , \infty \rbrace , \quad p =0, 1.
\end{equation}
Then, following our method, we obtain
\begin{align}
\mathfrak{g}^\infty_R &= \lbrace \lambda_{2l}, \; l=0, \ldots, \infty \rbrace \times V_0 \oplus \lbrace  \lambda_{2l+1},\; l=0, \ldots ,\infty  \rbrace \times V_1  , \nonumber \\
 \mathcal{I}&= \lbrace  \lambda_{2m}, \; m=2 , \ldots, \infty  \rbrace \times V_0 \oplus \lbrace  \lambda_{2m+1} , \; m =1, \ldots, \infty  \rbrace \times V_1   , \nonumber \\
 \check{\mathfrak{g}}_R&=\lbrace \lambda_0 , \lambda_2 \rbrace \times V_0 \oplus \lbrace \lambda_1 \rbrace \times V_1 .
\end{align}

If we now perform an identification and rename the generators of $\check{\mathfrak{g}}_R$ as follows:
\begin{align}
 J_{ab}&=\lambda_0\tilde{J}_{ab},\\
   Z_{ab}&=\lambda_2 \tilde{J}_{ab},  \\
  P_a&=\lambda_1 \tilde{P}_a, \\
 Q_\alpha &=\lambda_1 \tilde{Q}_\alpha ,
\end{align}
it is straightforward to show that we end up with a new superalgebra, $[\mathfrak{osp(2|1)}\otimes\mathfrak{sp(2)}]\ominus\mathcal{I}$, which corresponds to the non-standard Maxwell superalgebra \cite{Lukierski} and read
\begin{align}
   \left[J_{ab},J_{cd}\right]=&\eta_{bc}J_{ad}-\eta_{ac}J_{bd}-\eta_{bd}J_{ac}+\eta_{ad}J_{bc},\\
\left[J_{ab},Z_{cd}\right]=&\eta_{bc}Z_{ad}-\eta_{ac}Z_{bd}-\eta_{bd}Z_{ac}+\eta_{ad}Z_{bc},\\
  \left[J_{ab},P_c\right]=& \eta_{bc}P_a-\eta_{ac}P_b,\\
 \left[P_a,P_b\right]=& Z_{ab},\\
 \left[J_{ab},Q_\alpha\right]=& \frac{1}{2}\left(\Gamma_{ab}Q\right)_\alpha,\\ 
 \left[P_{a},Q_\alpha\right]=&0,\\
 \left[Z_{ab},Z_{cd}\right]=&0,\\
\left[Z_{ab},P_{c}\right]=&0,\\
\left[Z_{ab},Q_\alpha\right]=&0,\\
  \left\{Q_\alpha,Q_\beta\right\}=&-\frac{1}{2}\left(\Gamma^{ab}C\right)_{\alpha\beta}Z_{ab}.
\end{align}
We have thus reached the non-standard Maxwell superalgebra \cite{Lukierski} by performing an infinite $S$-expansion with ideal subtraction on the $AdS$ superalgebra. 

\paragraph{Invariant tensor of the non-standard Maxwell superalgebra in $D=3$}\

We write the components of the invariant tensor of the $AdS$ superalgebra (see Ref. \cite{Pat2}) as follows:
\begin{align}
 \langle \tilde{J}_{ab} \tilde{J}_{cd}\rangle=& \tilde{\mu}_{0}\left(\eta_{ad}\eta_{bc}-\eta_{ac}\eta_{bd}\right),\\
  \langle \tilde{J}_{ab} \tilde{P}_{c}\rangle=&\tilde{\mu}_{1} \epsilon_{abc},\\
 \langle \tilde{P}_{a}  \tilde{P}_{b}\rangle=&\tilde{\mu}_0\eta_{ab},\\
 \langle \tilde{Q}_\alpha \tilde{Q}_{\beta}\rangle=&\left(\tilde{\mu}_{0}-\tilde{\mu}_{1}\right)C_{\alpha\beta},
\end{align}
where $\tilde{\mu}_0$ and $\tilde{\mu}_1$ are arbitrary constants.
Using (\ref{topologicalInvariant}), the non-zero components of the invariant tensor for the non-standard Maxwell superalgebra can be written as

\begin{align}
 \langle J_{ab} J_{cd}\rangle=& \alpha^{0} \langle \tilde{J}_{ab} \tilde{J}_{cd}\rangle  =  \alpha^{0} \tilde{\mu}_{0}\left(\eta_{ad}\eta_{bc}-\eta_{ac}\eta_{bd}\right) \equiv \tilde{\alpha}^0 \left(\eta_{ad}\eta_{bc}-\eta_{ac}\eta_{bd}\right), \\
 \langle J_{ab} Z_{cd}\rangle=& \alpha^2 \langle \tilde{J}_{ab} \tilde{J}_{cd}\rangle = \alpha^2 \tilde{\mu}_{0}\left(\eta_{ad}\eta_{bc}-\eta_{ac}\eta_{bd}\right) \equiv \tilde{\alpha}^2  \left(\eta_{ad}\eta_{bc}-\eta_{ac}\eta_{bd}\right) ,\\
 \langle J_{ab}  P_{c}\rangle=&\alpha^{1} \langle \tilde{J}_{ab} \tilde{P}_{c}\rangle = \alpha^1 \tilde{\mu}_{1} \epsilon_{abc} \equiv \tilde{\alpha}^1 \epsilon_{abc},\\
 \langle P_{a}  P_{b}\rangle=& \alpha^2 \langle \tilde{P}_{a}  \tilde{P}_{b}\rangle= \alpha^{2}\tilde{\mu}_0\eta_{ab} \equiv \tilde{\alpha}^2 \eta_{ab},\\
 \langle Q_\alpha Q_{\beta}\rangle=& \alpha^{2} \tilde{\mu}_{0} C_{\alpha\beta} \equiv \tilde{\alpha}^2 C_{\alpha \beta},
\end{align}

where $\alpha^m$, $m=0,1,2$, are arbitrary constants, and where we have defined 
\begin{align}
 \tilde{\alpha}^0 \equiv & \alpha^0 \tilde{\mu}_0 , \\  \tilde{\alpha}^1 \equiv & \alpha^1 \tilde{\mu}_1 , \\
  \tilde{\alpha}^2 \equiv & \alpha^2 \tilde{\mu}_0.
\end{align}

\subsection{PP-wave (super)algebra in $D=11$ from the $AdS_4\times S^{7}$ (super)algebra}

The (super-)PP-wave algebra in eleven dimensions can be obtained through a generalized In\"on\"u-Wigner contraction of the $AdS_4\times S^{7}$ (super)algebra in $D=11$ (see Ref. \cite{ppwave1}). \footnote{This can be viewed as the Penrose limit of the $AdS \times S$ metric.}

In the following, we will reach the same result by performing an infinite $S$-expansion with ideal subtraction. Indeed, as shown in Section \ref{Generalized}, our prescription of infinite $S$-expansion with ideal subtraction is able to reproduce a generalized In\"{o}n\"{u}-Wigner contraction.

\subsubsection{PP-wave algebra from the $AdS_4\times S^{7}$ algebra}

The $AdS_4\times S^7$ algebra in $D=11$ can be written in the following traditional form (see Ref. \cite{ppwave1}):
\begin{align}
\left[P_a,P_b\right]=&4J_{ab}, & \left[P_{a'},P_{b'}\right]=&4J_{a'b'},\\
\left[J_{ab},P_c\right]=&2\eta_{bc}P_a, & \left[J_{a'b'},P_{c'}\right]=&2\eta_{b'c'}P_{a'},\\
 \left[J_{ab},J_{cd}\right]=&4\eta_{ad}J_{bc}, & \left[J_{ab},J_{cd}\right]=&4\eta_{ad}J_{bc},
\end{align}
where the vector index of $AdS_4$ is $i=0,1,2,3$, and that of $S^7$ is $a'=4,5,6,7,8,9,\sharp$.
 
Let us follow Ref. \cite{ppwave1} and define the light cone components of the momenta $P$'s and boost generators $P^*$'s as
\begin{align}
 P_{\pm}\equiv &\frac{1}{\sqrt{2}}\left(P_{\natural}\pm P_0 \right),\\
 P_m=&\left(P_i,P_{i'}\right),\\
 P_m^*=&\left(P^*_{i}=J_{i0},P^*_{i'}=J_{i'\natural}\right).
\end{align}
Thus, we can write the commutation relations of $AdS_4\times S^7$ as follows:
\begin{align}
 \left[P_i,P_+\right]=&2\sqrt{2}P_i^*, & \left[P_{i'},P_{+}\right]=&-\frac{1}{\sqrt{2}}P_{i}^*,\nonumber\\
 \left[P_{i}^*,P_{+}\right]=&-\frac{1}{\sqrt{2}}P_{i}, &\left[P_{i'}^*,P_{+}\right]=&\frac{1}{\sqrt{2}}P_{i'}, \nonumber\\
 \left[P_{i},P_{-}\right]=&-2\sqrt{2}P_{i}^*, &\left[P_{i'},P_{-}\right]=&-\frac{1}{\sqrt{2}}P_{i'}^* ,\nonumber\\
 \left[P_{i}^*,P_{-}\right]=&\frac{1}{\sqrt{2}}P_i, & \left[P_{i'}^*,P_{-}\right]=&\frac{1}{\sqrt{2}}P_{i'},\nonumber\\
 \left[P_{i}^*,P_{j}\right]=&-\frac{1}{\sqrt{2}}\eta_{ij}\left(P_- -P_+\right), & \left[P_{i'}^*,P_{j'}\right]=&-\frac{1}{\sqrt{2}}\eta_{i'j'}\left(P_-+P_+\right),\label{ppwavecommutators}\\
 \left[P_{i},P_{j}\right]=&4J_{ij}, & \left[P_{i'},P_{j'}\right]=&-J_{i'j'},\nonumber\\
 \left[P_{i}^*,P_{j}^*\right]=&J_{ij}, & \left[P_{i'}^*,P_{j'}^*\right]=&-J_{i'j'} ,\nonumber\\
\left[J_{ij},P_{k}\right]=&2\eta_{jk}P_{i}, & \left[J_{i'j'},P_{k'}\right]=&2\eta_{j'k'}P_{i'},          \nonumber \\
\left[J_{ij},P_{k}^*\right]=&2\eta_{jk}P_{i}^*, & \left[J_{i'j'},P_{k'}^*\right]=&2\eta_{j'k'}P_{i'}^*,      \nonumber \\
\left[J_{ij},J_{kl}\right]=&4\eta_{il}J_{jk}, &     \left[J_{i'j'},J_{k'l'}\right]=&4\eta_{i'l'}J_{j'k'}. \nonumber
\end{align}
Before applying the infinite $S$-expansion method with subsequent ideal subtraction, we consider the following subspaces partition of the $AdS_4\times S^7$ algebra:
\begin{align}
 V_0=&\left\{P_-,J_{ij},J_{i'j'}\right\},\nonumber\\
 V_1=&\left\{P_i,P_{i'},P_i^*,P_{i'}^*\right\},\label{spacesPP}\\
 V_2=&\left\{P_+ \right\}.\nonumber
 \end{align}

In this scenario, the Weimar-Woods conditions hold. 
Thus, following the procedure developed in Subsection \ref{infexp} in the case in which the original algebra satisfy the Weimar-Woods conditions, we now perform an infinite $S$-expansion with the semigroup $S^{(\infty)}$ on the algebra $\mathfrak{g}=AdS_4\times S^{7}$, considering the following resonant subset decomposition of the semigroup $S^{(\infty)}$:
\begin{align}
 S_0=&\left\{\lambda_0\right\}\cup\left\{\lambda_{1}, \ldots , \lambda_\infty \right\},\nonumber\\
S_1=&\left\{\lambda_1\right\}\cup\left\{\lambda_{2}, \ldots , \lambda_\infty \right\} ,\label{groupPP}\\
S_2=&\left\{\lambda_2\right\}\cup\left\{\lambda_{3}, \ldots , \lambda_\infty \right\}.\nonumber
\end{align}
The resonant subalgebra $\mathfrak{g}^\infty_R=\check{\mathfrak{g}}_R \oplus \hat{\mathfrak{g}}^\infty_R$ results to be given by the direct sum of the following terms:
\begin{align}
 S_0\times V_0&=\lbrace\lambda_0 \rbrace \times V_0 \oplus \left\{\lambda_{1} , \ldots , \lambda_\infty\right\} \times V_0,\\
 S_1\times V_1&= \lbrace\lambda_1 \rbrace \times  V_1 \oplus \left\{\lambda_{2}, \ldots , \lambda_\infty \right\} \times V_1,\\
 S_2\times V_2&= \lbrace \lambda_2 \rbrace \times V_2 \oplus \left\{\lambda_3 , \ldots , \lambda_\infty \right\} \times V_2.
\end{align}

Then, the ideal reads
\begin{align}
 \mathcal{I}=&\left(\left\{\lambda_{1} , \ldots , \lambda_\infty\right\}\times V_0\right)\oplus\left(\left\{\lambda_{2} , \ldots , \lambda_\infty\right\}\times V_1\right)\oplus\left(\left\{\lambda_{3} , \ldots , \lambda_\infty\right\}\times V_2\right),
\end{align}
and we can perform the ideal subtraction on the resonant subalgebra $\mathfrak{g}^\infty_R$.

Let us notice that the most relevant step consists in writing the commutation relations
\begin{align}
 \left[\lambda_1P_{i}^*,\lambda_1P_{j}\right]=&-\frac{1}{\sqrt{2}}\eta_{ij}\left(\lambda_2P_--\lambda_2P_+\right),\\
 \left[\lambda_1P_{i'}^*,\lambda_1P_{j'}\right]=&-\frac{1}{\sqrt{2}}\eta_{i'j'}\left(\lambda_2P_--\lambda_2P_+\right),
\end{align}
where $\lambda_2P_-$ belongs to the ideal, while $\lambda_2P_+$ does not.

Thus, after the ideal subtraction, we end up with the algebra $\mathfrak{g}_{PP}$ (the PP-wave algebra):
\begin{align}
\mathfrak{g}_{PP}\equiv \check{\mathfrak{g}}_R=\mathfrak{g}_R^\infty\ominus\mathcal{I}.
\end{align}
Indeed, if we now rename the generators as follows: $\lambda_0J_{ij}=\tilde{J}_{ij}$, $\lambda_0J_{i'j'}=\tilde{J}_{i'j'}$, $\lambda_0P_-=\tilde{P}_-$, $ \lambda_1P_{i}=\tilde{P}_i$,  $\lambda_1P_{i'}=\tilde{P}_{i'}$, $\lambda_1P_{i}^*=\tilde{P}_{i}^*$, $\lambda_1P_{i'}^*=\tilde{P}_{i'}^*$, $\lambda_2P_+=\tilde{P}_+$, we can finally write the PP-wave algebra in terms of the following commutation relations (we adopt the notation of Ref. \cite{ppwave1}):
\begin{align}
 \left[\tilde{P}_i,\tilde{P}_-\right]=&-2\sqrt{2}\tilde{P}_i^*, & \left[\tilde{P}_{i'},\tilde{P}_-\right]=&-\frac{1}{\sqrt{2}}\tilde{P}_{i}^*, \nonumber \\
 \left[\tilde{P}_m^*,\tilde{P}_{-}\right]=&-\frac{1}{\sqrt{2}}\tilde{P}_{m}, & \left[\tilde{P}_{m}^*,\tilde{P}_{n}\right]=&-\frac{1}{\sqrt{2}}\eta_{mn}\tilde{P}_{+}, \nonumber\\
 \left[\tilde{J}_{mn},\tilde{J}_{pq}\right]=&\;4\eta_{mq}\tilde{J}_{np}, & \left[\tilde{J}_{mn},\tilde{P}_{p}^*\right]=&\;2\eta_{np}\tilde{P}^*_m , \nonumber \\
 \left[\tilde{J}_{mn},\tilde{J}_{pq}\right]=&\;4\eta_{mq}\tilde{J}_{np}. 
\end{align}
We have thus reached the PP-wave algebra in $D=11$ with an infinite $S$-expansion with ideal subtraction, starting from the $AdS_4\times S^{7}$ algebra, reproducing, in this way, a generalized In\"on\"u-Wigner contraction.

\subsubsection{Super-PP-wave algebra from the $AdS_4\times S^{7}$ superalgebra}

We now extend the previous analysis to the supersymmetric case.
We thus consider the addition of fermionic generators to the $AdS_4\times S^{7}$ algebra, and then perform an infinite $S$-expansion with consequent ideal subtraction. 

The commutators involving the supercharges $Q$'s can be decomposed in the following way:
\begin{equation}
Q =Q_+ + Q_-, \;\;\; Q_\pm = Q_\pm\mathcal{P}_\pm ,
\end{equation}
using the light cone projection operators
\begin{align}
 \mathcal{P}_\pm=&\frac{1}{2}\Gamma_\pm\Gamma_\mp,\\
 \Gamma_\pm\equiv&\frac{1}{\sqrt{2}}\left(\Gamma_\natural\pm\Gamma_0\right),
\end{align}
where the $\Gamma$'s are the gamma matrices in eleven dimensions.
Now we add $Q_\pm$ into the subspaces partition, before applying the ideal subtraction, and we thus write
\begin{align}
 V_0=&\left\{P_-,J_{mn},Q_-\right\},\nonumber\\
 V_1=&\left\{P_m,P_{m}^*,Q_+\right\},\label{distribucion}\\
 V_2=& \left\{P_+ \right\}.\nonumber
 \end{align}
In this way, proceeding analogously to the previous (bosonic) case,
performing an infinite $S$-expansion involving the infinite abelian semigroup $S^{(\infty)}$ and subsequently subtracting the ideal
\begin{align}
 \mathcal{I}=&\left(\left\{\lambda_{1} , \ldots , \lambda_\infty\right\}\times V_0\right)\oplus\left(\left\{\lambda_{2} , \ldots , \lambda_\infty\right\}\times V_1\right)\oplus\left(\left\{\lambda_{3} , \ldots , \lambda_\infty\right\}\times V_2\right),
\end{align}
we end up with the following commutation relations:
\begin{align}
\left[\lambda_0P_-,\lambda_1Q_+\right]=&-\frac{3}{2\sqrt{2}}\lambda_{1}Q_+I,\\
\left[\lambda_0P_{-},\lambda_0Q_-\right]=&-\frac{1}{2\sqrt{2}}\lambda_{0}Q_-I\\
  \left[\lambda_1P_i,\lambda_0Q_-\right]=&\frac{1}{\sqrt{2}}\lambda_{1}Q_+\Gamma^-I\Gamma_i,\\ 
  \left[\lambda_1P_{i'},\lambda_0Q_-\right]=&\frac{1}{2\sqrt{2}}\lambda_{1}Q_+\Gamma^-I\Gamma_{i'},\\
\left[\lambda_1P^*_{m},\lambda_0Q_-\right]=&\frac{1}{2\sqrt{2}}\lambda_{1}Q_+\Gamma_m\Gamma^-,\\ 
\left[\lambda_0J_{mn},\lambda_0Q_\pm\right]=&\frac{1}{2}\lambda_{0}Q_\pm\Gamma_{mn},\\
\left\{\lambda_1Q_+,\lambda_1Q_+\right\}=&-2\mathcal{C}\Gamma^+\lambda_{1}P_+,\\
\left\{\lambda_0 Q_-,\lambda_0 Q_-\right\}=&-2\mathcal{C}\Gamma^- \lambda_0 P_- - \sqrt{2}\mathcal{C}\Gamma^-I\Gamma^{ij}\lambda_{0}J_{ij}+\frac{1}{\sqrt{2}}\mathcal{C}\Gamma^-I\Gamma^{i'j'}\lambda_{0}J_{i'j'},\\
\left\{\lambda_1Q_+,\lambda_0Q_-\right\}=&\left(-2\mathcal{C}\Gamma^m\lambda_{1}P_m-4\mathcal{C}I\Gamma^i\lambda_{1}P^*_i-2\mathcal{C}I\Gamma^{i'}\lambda_{1}P_{i'}^*\right)\lambda_0 \mathcal{P}_-.
\end{align}
Then, we only have to properly rename the generators in order to find the algebra
\begin{align}
\left[\tilde{P}_-,\tilde{Q}_+\right]=&-\frac{3}{2\sqrt{2}} \tilde{Q}_+I,\nonumber\\
\left[\tilde{P}_{-},\tilde{Q}_-\right]=&-\frac{1}{2\sqrt{2}}\tilde{Q}_-I\nonumber\\
  \left[\tilde{P}_i,\tilde{Q}_-\right]=&\frac{1}{\sqrt{2}} \tilde{Q}_+\Gamma^-I\Gamma_i,\nonumber\\ 
  \left[\tilde{P}_{i'},\tilde{Q}_-\right]=&\frac{1}{2\sqrt{2}} \tilde{Q}_+\Gamma^-I\Gamma_{i'},\nonumber\\
\left[\tilde{P}^*_{m},\tilde{Q}_-\right]=&\frac{1}{2\sqrt{2}} \tilde{Q}_+\Gamma_m\Gamma^-, \label{ffpp}\\ 
\left[\tilde{J}_{mn},\tilde{Q}_\pm\right]=&\frac{1}{2}\tilde{Q}_\pm\Gamma_{mn},\nonumber\\
\left\{\tilde{Q}_+,\tilde{Q}_+\right\}=&-2\mathcal{C}\Gamma^+ \tilde{P}_+,\nonumber\\
\left\{\tilde{Q}_-,\tilde{Q}_-\right\}=&-2\mathcal{C}\Gamma^-\tilde{P}_- - \sqrt{2}\mathcal{C}\Gamma^-I\Gamma^{ij}\tilde{J}_{ij}+ \frac{1}{\sqrt{2}}\mathcal{C}\Gamma^-I\Gamma^{i'j'}\tilde{J}_{i'j'},\nonumber\\
\left\{\tilde{Q}_+,\tilde{Q}_-\right\}=&\left(-2\mathcal{C}\Gamma^m \tilde{P}_m-4\mathcal{C}I\Gamma^i \tilde{P}^*_i-2\mathcal{C}I\Gamma^{i'} \tilde{P}_{i'}^*\right)\tilde{\mathcal{P}}_-.\nonumber
\end{align}
These commutation relations correspond to those of the super-PP-wave algebra (we have adopted the same notation of Ref. \cite{ppwave1}).

We have thus reached the super-PP-wave algebra in $D=11$ by performing an infinite $S$-expansion with ideal subtraction, starting from the super-$AdS_4\times S^{7}$, reproducing, in this way, a generalized In\"on\"u-Wigner contraction in the case of a supersymmetric algebra.

\section{Comments and possible developments} \label{Comments}

The $S$-expansion method has the peculiarity of being able to reproduce the standard In\"on\"u-Wigner contraction as a special case of the procedure called $0_S$-reduction (see Ref. \cite{Izaurieta}).

Furthermore, as shown in Ref. \cite{Izaurieta}, the information on the subspace structure of the original (super)algebra can be used in order to find resonant subalgebras of the $S$-expanded
(super)algebra and, by extracting reduced algebras from
the resonant subalgebra, one can reproduce the generalized In\"on\"u-Wigner contraction within this scheme.

In the present work, we have given a new prescription for $S$-expansion, based on using an \textit{infinite} abelian semigroup $S^{(\infty)}$ and by performing the subtraction of an infinite ideal subalgebra. We have explicitly shown that the subtraction of the infinite ideal subalgebra corresponds to a reduction, leading to a reduced (super)algebra.
In particular, it can be viewed as a (generalization of the) $0_S$-reduction, since it reproduces the same effects. This method can be also interpreted as a different, alternative view of the generalized In\"on\"u-Wigner contraction. Indeed, our ``infinite $S$-expansion" procedure allows to reproduce the standard as well as the generalized In\"on\"u-Wigner contraction. The removal of the infinite ideal is crucial, since it allows to end up with finite dimensional Lie (super)algebras.

This infinite $S$-expansion procedure represents an extension and generalization of the finite one, allowing a deeper view on the maps linking different algebras.

We have then explained how write the invariant tensors for the (super)algebras which can be obtained by applying our method of infinite $S$-expansion with ideal subtraction.
This procedure allows to develop the dynamics and construct the Lagrangians of physical theories. 
In particular, in this context the construction of Chern-Simons forms becomes more accessible, and it would be particularly interesting to develop them in (super)gravity theories in higher dimensions, by following our approach.

This work reproduces the results already presented in the literature, concerning expansions and contractions of Lie (super)algebras, and also gives some new features. 
Moreover, it gives further connections between the contraction processes and the expansion methods, which was an open question already mentioned in Ref. \cite{Khasanov}.

On the other hand, this paper can also contribute to understand the nature of the semigroups used in the $S$-expansion processes, since a contraction is usually applied over a physical constant. 

In our work, we have restricted our study to the cases involving an infinite semigroup $S^{(\infty)}$ related to the set $(\mathbb{N},+)$. We leave a possible extension to the set $(\mathbb{Z},+)$ to future works. 
This further analysis would be interesting, since it would produce an $S$-expansion involving an abelian group (with respect to the sum operation), rather than a semigroup.

Another possible development of our work would consists in applying the method developed in the present paper to the case studied in Ref. \cite{krishnan}, where the authors showed that interpreting the inverse $AdS_3$ radius $1/l$ as a Grassmann variable results in a formal map from gravity in $AdS_3$ to gravity in flat space. 
The underlying reason for this relies in the In\"{o}n\"{u}-Wigner contraction. 
They systematically developed the possibility that In\"{o}n\"{u}-Wigner contraction could be turned into an algebraic operation through a Grassmann approach. 

We argue that this could present a strong chance that a generalization of their algebraic approach should also work in our case as well, and it would thus be interesting to extend our method by exploiting a Grassmann-like approach. Some work is in progress on this topic.

\section*{Acknowledgments}

We are grateful to L. Andrianopoli, R. D'Auria, and M. Trigiante for the stimulating discussions and the constant support. 
The authors also wish to thank F. Lingua, R. Caroca, N. Merino, P.K. Concha, and E.K. Rodr\'{i}guez for the enlightening suggestions. 
One of the authors (D. M. Pe\~{n}afiel) was supported by grants from the \textit{Comisi\'{o}n Nacional de Investigaci\'{o}n Cient\'{i}fica y Tecnol\'{o}gica} CONICYT and from the \textit{Universidad de Concepci\'{o}n}, Chile, and wishes to acknowledge L. Andrianopoli, R. D'Auria, and M. Trigiante for their kind hospitality at DISAT - \textit{Dipartimento di Scienza Applicata e Tecnologia} of the \textit{Polytechnic of Turin}, Italy.


\end{document}